\documentclass[pra,aps,superscriptaddress,twocolumn,nofootinbib,longbibliography,accepted=2023-12-18]{quantumarticle}
\pdfoutput=1

\usepackage{amsfonts,amsmath,amssymb,amsthm,bbm,bm,centernot,color,enumerate,enumitem,graphicx,hyperref,mathdots,mathtools,multirow,pgfplots,soul,subfigure,thmtools,tikz}
\usepackage[capitalise, noabbrev]{cleveref}
\usepackage[numbers]{natbib}
\usepackage[T1]{fontenc}

\hypersetup{colorlinks=true,linkcolor=blue,citecolor=blue,urlcolor=blue}
\pgfplotsset{compat=1.18} 

\def\E{ {\cal E} }
\def\F{ {\cal F} } 
\def\T{ {\cal T} }

\def\>{\rangle}
\def\<{\langle}

\newcommand{\ket}[1]{| {#1} \rangle}
 
\newcommand{\ketbra}[2]{\ensuremath{\left|#1\right\rangle\!\!\left\langle#2\right|}}

\newcommand{\tr}[1]{\mathrm{Tr}\left( #1 \right)}
\newcommand{\trr}[2]{\mathrm{Tr}_{#1}\left( #2 \right)}
\newcommand{\iden}{\mathbbm{1}}

\newcommand{\conv}[1]{\mathrm{conv}\left[#1\right]}

\renewcommand{\v}[1]{\ensuremath{\boldsymbol #1}}

\theoremstyle{plain}
\newtheorem{thm}{Theorem}
\newtheorem{lem}[thm]{Lemma}
\newtheorem{prop}[thm]{Proposition}
\newtheorem{cor}[thm]{Corollary}
\theoremstyle{definition}

\newtheorem{rem}[thm]{Remark}
\newtheorem{ex}[thm]{Example}

\begin{document}

\title{Resource engines}
	
	\author{Hanna Wojew{\'o}dka-{\'S}ci\k{a}{\.z}ko}	
	\affiliation{Institute of Mathematics, University of Silesia in Katowice, Bankowa 14, 40-007 Katowice, Poland}
	\affiliation{Institute of Theoretical and Applied Informatics, Polish Academy of Sciences,  Ba{\l}tycka 5, 44-100 Gliwice, Poland}
	\email{hanna.wojewodka@us.edu.pl}
    
    \author{Zbigniew Pucha\l{}a}
    \affiliation{Institute of Theoretical and Applied Informatics, Polish Academy of Sciences,  Ba{\l}tycka 5, 44-100 Gliwice, Poland}
	
	\author{Kamil Korzekwa}
	\affiliation{Faculty of Physics, Astronomy and Applied Computer Science, Jagiellonian University, 30-348 Krak\'{o}w, Poland}
	
\date{\today}

\begin{abstract}
    In this paper we aim to push the analogy between thermodynamics and quantum resource theories one step further. Previous inspirations were based predominantly on thermodynamic considerations concerning scenarios with a single heat bath, neglecting an important part of thermodynamics that studies heat engines operating between two baths at different temperatures. Here, we investigate the performance of resource engines, which replace the access to two heat baths at different temperatures with two arbitrary constraints on state transformations. The idea is to imitate the action of a two--stroke heat engine, where the system is sent to two agents (Alice and Bob) in turns, and they can transform it using their constrained sets of free operations. We raise and address several questions, including whether or not a resource engine can generate a full set of quantum operations or all possible state transformations, and how many strokes are needed for that. We also explain how the resource engine picture provides a natural way to fuse two or more resource theories, and we discuss in detail the fusion of two resource theories of thermodynamics with two different temperatures, and two resource theories of coherence with respect to two different bases.
\end{abstract}

\maketitle


\section{Introduction}

Thermodynamics occupies a distinguished place among physical theories: it permeates all other fields and finds applicability everywhere, from astrophysics~\cite{davies1978thermodynamics} and biophysics~\cite{zuckerman2010statistical} to condensed matter~\cite{lifshitz2013statistical} and physics of computation~\cite{bennett1982thermodynamics}. The reason for this is that thermodynamics actually forms a meta-theory that employs statistical reasoning to tell us which state transformations are allowed, and which are probabilistically impossible. In other words, thermodynamics can be seen as a field studying the accessibility and inaccessibility of one physical state from another~\cite{giles1964mathematical}, abstracting away whether these states describe electromagnetic radiation, a gas of spin particles or electrons in an atom. 

This perspective inspired the development of quantum resource theories~\cite{chitambar2019quantum}, where one investigates allowed transformations between quantum states under general, not necessarily thermodynamic, constraints. The best-known example is given by the theory of entanglement~\cite{horodecki2009quantum}, where bipartite system transformations are constrained to local operations and classical communication, and one is interested in the possible manipulations of entangled states. Other widely studied examples include state transformations under operations that are incoherent~\cite{baumgratz2014quantifying}, symmetric~\cite{marvianthesis} or Clifford~\cite{veitch2014resource}. In all these theories, clear parallels with thermodynamic considerations have been made. Similarly as every thermodynamic transition can be made possible by investing enough work, universal resources in other theories (such as ebits~\cite{bennett1996concentrating}, refbits~\cite{van2005quantifying} or coherent bits~\cite{chitambar2016relating}) have also been identified and their optimal manipulations have been investigated. Also, the thermodynamic concept of catalysis was extended to general resource theories, where catalysts are given by ancillary systems that make otherwise forbidden transitions possible, while being unchanged in the process~\cite{jonathan1999entanglement,bu2016catalytic}. Finally, probably the most important link between thermodynamics and other resource theories is through the second law of thermodynamics~\cite{horodecki2002laws}, where general resource monotones~\cite{gonda2019monotones}, such as entropy of entanglement~\cite{horodecki2009quantum} or relative entropy of coherence~\cite{baumgratz2014quantifying}, play the role of monotonically decreasing thermodynamic free energy, and the questions concerning reversibility and irreversibility of a given resource theory are central to the field~\cite{brandao2008entanglement,kumagai2013entanglement,korzekwa2019avoiding,lami2023no}.

Given how fruitful the thermodynamic inspirations have been so far for quantum resource theories, in this paper we aim at pushing this analogy one step further. So far most of such inspirations have been based on thermodynamic considerations concerning scenarios with a single heat bath. Notable exception is given by the resource theory thermodynamics itself, where the efficiency and performance of quantum heat engines~\cite{ng2017surpassing,tajima2017finite,bera2021attaining} and more general autonomous thermal machines~\cite{tonner2005autonomous,mitchison2019quantum} working between two baths at different temperatures were studied from a resource-theoretic perspective. However, the approach to general resource theories neglects this very important part of thermodynamics with access to two heat baths. It is true that, from a purely resource-theoretic perspective, having access to two infinite baths at different temperatures in some sense trivialises the theory, as one can then perform infinite amount of work, and so every state transformation becomes possible (at least in the semi-classical regime of energy-incoherent states~\cite{lostaglio2015description}). However, the physics of heat engines is far from being trivial, one just needs to ask different questions. Instead of looking for the amount of work that can be extracted from a given non-equilibrium state, one rather asks about the optimal efficiency of converting heat into work, or about the maximum power of a heat engine. 

We thus propose to investigate the performance of \emph{resource engines}, which generalise the concept of heat engines by replacing the access to two heat baths at different temperatures with two arbitrary constraints on state transformations. More precisely, we consider two agents (traditionally refereed to as Alice and Bob), each of which is facing a different constraint, meaning that each of them can only prepare a subset of free states, $F_A$ and $F_B$, and can only perform quantum operations from a subset of free operations, $\F_A$ and $\F_B$. Now, the idea is to imitate the action of a two-stroke heat engine: instead of subsequently connecting the system to the hot and cold bath, it is sent to Alice and Bob in turns and they can perform any operation on it from their constrained sets $\F_A$ and $\F_B$. Since the free states and operations of Alice will generally be resourceful with respect to Bob's constraints (and vice versa), a number of such communication rounds with local constrained operations (i.e., strokes of a resource engine) may generate  quantum states outside of $F_A$ and $F_B$. Thus, by fusing two resource theories described by $(F_A,\F_A)$ and $(F_B,\F_B)$, one can obtain a new resource theory with free operations $\F_{AB}\supseteq \F_A \cup \F_B$ and free states $F_{AB}\supseteq F_A \cup F_B$.

A number of natural questions then arise. First, can a~resource engine defined by given two constraints generate a full set of quantum operations, or at least approach every element of this set arbitrarily well with the number of strokes going to infinity? Alternatively, can it achieve all possible final states starting from states belonging to $F_A$ or $F_B$? If the answer to these questions is yes, then can we bound the number of strokes needed to generate every operation or state? And if there exists a state that is maximally resourceful with respect to both Alice's and Bob's constraints, what is the minimal number of strokes needed to create it? Note that given that each stroke takes a fixed amount of time, this effectively corresponds to studying the optimal power of a~resource engine. One can also ask about the equivalent of engine's efficiency. Namely, whenever Bob gets a state from Alice and transforms it using an operation from $\F_B$, he necessarily decreases the resource content of the state with respect to his constraint, but may increase it with respect to Alice's constraint. Thus, one may investigate the optimal trade-off, i.e., the efficiency of transforming his resource into Alice's resource.

In this paper, we start with our resource-theoretic perspective on standard heat engines in Sec.~\ref{sec:thermo}, where Alice and Bob are constrained to having access to heat baths at different temperatures. More precisely, in Sec.~\ref{sec:thermo_setting}, we recall the necessary notions of thermal operations and thermomajorisation, set the notation, and formally state the main problems of athermality engines we want to investigate. Next, in Sec.~\ref{sec:thermo_qubit}, we fully solve these problems for an elementary example of a two-level system. Section~\ref{sec:thermo:bounding} contains our main results for athermality engines with arbitrary $d$-dimensional systems, where we lower and upper bound the set of achievable states $F_{AB}$, as well as find its exact form in the limit of the hotter bath having infinite temperature. We then proceed to Sec.~\ref{sec:coherence} that is devoted to the concept of unitary coherence engines, where Alice and Bob are constrained to only performing unitary operations diagonal in their fixed bases, so that coherence with respect to these bases is a~resource for them. First, in Sec.~\ref{sec:coherence_setting}, we set the scene by recalling the necessary formalism, fixing the notation and stating the problems. Next, as in the athermality case, we fully solve these problems for the simplest case of a two-level system in Sec.~\ref{sec:coherence_qubit}. In Sec.~\ref{sec:coherence_operations}, we then derive and discuss the conditions under which the full set of unitary operations can be performed jointly by Alice and Bob, i.e., when $\F_{AB}$ becomes the full set of unitary operations. We analyse the number of strokes $N$ needed to get all these operations in Sec.~\ref{sec:coherence_bounds}, presenting both lower and upper bounds for $N$. Finally, in Sec.~\ref{sec:coherence_optimal}, we discuss the problem of using the resource engine to produce a state that is simultaneously maximally resourceful for both Alice and Bob. We end the paper in Sec.~\ref{sec:outlook}, where we discuss relations between resource engines and known problems within quantum information, and outline opportunities for future research.


\section{Athermality engine}
\label{sec:thermo}


\subsection{Setting the scene}
\label{sec:thermo_setting}

In the resource-theoretic approach to thermodynamics~\cite{horodecki2013fundamental}, one assumes there is a single agent $A$ that has access to a thermal heat bath at inverse temperature $\alpha=1/(k T_A)$, with $k$ denoting the Boltzmann constant, and is allowed to unitarily couple the system of interest to the bath in an energy-conserving way. More formally, given a quantum system in a state $\rho$ and described by a Hamiltonian $H$, the set of operations $\F_A$ that the agent can perform consists of \emph{thermal operations} $\E$ with respect to $\alpha$ defined by~\cite{janzing2000thermodynamic,horodecki2013fundamental}
\begin{equation}
    \label{eq:thermal_op} 
    \E(\rho)=\trr{E}{U(\rho\otimes\gamma_E)U^\dagger},
\end{equation}
where 
\begin{equation}
    \label{eq:gibbs}
    \gamma_E=\frac{\exp(-\alpha H_E)}{\tr{\exp(-\alpha H_E)}}
\end{equation}
is a thermal Gibbs state of the environment described by an arbitrary Hamiltonian $H_E$, and $U$ is a joint system-bath unitary preserving the total energy, i.e.,
\begin{equation}
    [U,H\otimes \iden_E + \iden\otimes H_E]=0.
\end{equation}
The central question is then given by the following interconversion problem: which final states $\sigma$ are achievable from a given initial state $\rho$ via thermal operations? Note that the only free state that any state can be mapped to is the system's thermal Gibbs state $\gamma$ (defined by Eq.~\eqref{eq:gibbs} with $H_E$ replaced by $H$), and all other states are treated as resources. In other words, the free set $F_A$ consists of a single element $\gamma$.

Here, we will focus on the quasi-classical version of the interconversion problem with initial and final states of the system commuting with $\gamma$. This case is usually referred to as incoherent thermodynamics, since one then focuses on interconversion between states that do not have quantum coherence between different energy eigensectors. Instead of representing quantum states as density matrices of size $d$, one can then use a simpler representation using $d$-dimensional probability distributions describing occupations of different energy levels. We will thus denote the initial, final and Gibbs states by probability vectors $\v{p}$, $\v{q}$ and $\v{\gamma}$, respectively. 

The main reason why we restrict to the incoherent setting is because then the necessary and sufficient conditions for the existence of a thermal operation mapping $\v{p}$ to $\v{q}$ are known to be fully characterised by a thermomajorisation condition $\v{p}\succ_{\v{\gamma}} \v{q}$~\cite{horodecki2013fundamental} (also known as $d$-majorisation~\cite{ruch1980generalization}). For a full review of this subject, we refer the reader to Ref.~\cite{lostaglio2017thermodynamic}, while in Appendix~\ref{app:thermo} we summarise the bits of the theory necessary for our purposes. Here, we only note that a thermomajorisation curve (also known as the Lorenz curve) of $\v{p}$ with respect to $\v{\gamma}$ is a piece-wise linear and concave curve on a plane starting at the origin and consisting of segments with the horizontal $x$ length given by $\gamma_i$ and vertical $y$ lengths given by $p_i$ (note that concavity of the curve enforces a particular ordering of these segments). Then, we say that $\v{p}$ thermomajorises $\v{q}$, $\v{p}\succ_{\v{\gamma}} \v{q}$, if and only if the thermomajorisation curve of $\v{p}$ lies above that of $\v{q}$ everywhere. 

Now, our aim is to study a modified setting with two agents, $A$ and $B$, that can exchange the processed system between each other, and with each of the agents being constrained to only performing thermal operation with respect to their (unequal) temperatures\footnote{Note the difference with a typical approach to autonomous thermal machines, where the system is constantly connected to two baths, but the interaction with them is constrained (fixed). Here, the system interacts only with one bath at a time, but its interaction is unconstrained (i.e., any thermal operation on the system can be performed).}. More precisely, we consider the set of free operations $\F_A$ to be given by thermal operations with inverse temperature $\alpha$ and the corresponding free state given by $\v{\gamma}$ with
\begin{equation}
    \gamma_k = \frac{e^{-\alpha E_k}}{Z_\alpha},\qquad Z_\alpha= \sum_{i=1}^d e^{-\alpha E_{i}},
\end{equation}
where $\{ E_i\}$ denotes the energy levels of the system; and the free set of operations $\F_B$ to be given by thermal operations with inverse temperature $\beta<\alpha$ and the corresponding free state given by $\v{\Gamma}$ with
\begin{align}
    \Gamma_k &= \frac{e^{-\beta E_k}}{Z_\beta},\qquad Z_\beta= \sum_{i=1}^d e^{-\beta E_{i}}.
\end{align}

The main question that we will investigate is: what is the resulting set $F_{AB}$ of achievable states if each of the agents is constrained to their own set of free states and operations (so that together they can generate any operation from the set $\F_{AB}$)? In other words, we look for all states $\v{p}^{(N-1)}$ (and their convex combinations via randomised strategies) that can be obtained via a sequence of thermomajorisations either from $\v{\gamma}$ or $\v{\Gamma}$, i.e., that satisfy one of the following: 
\begin{subequations}
\begin{align}
    \v{\gamma}&=:\v{p}^{(0)} \succ_{\v{\Gamma}} \v{p}^{(1)} \succ_{\v{\gamma}} \v{p}^{(2)} \succ_{\v{\Gamma}} \v{p}^{(3)} \succ_{\v{\gamma}}\dots \v{p}^{(N-1)},\\
    \v{\Gamma} &=:\v{p}^{(0)}\succ_{\v{\gamma}} \v{p}^{(1)} \succ_{\v{\Gamma}} \v{p}^{(2)} \succ_{\v{\gamma}} \v{p}^{(3)} \succ_{\v{\Gamma}}\dots \v{p}^{(N-1)},
\end{align}
\end{subequations}
for a given value of $N$ (set of states $F_{AB}^{(N)}$ achievable after $N$ strokes) or for $N\to\infty$ (set $F_{AB}$ of all achievable states). Note that in the investigated setting we assume that the agents exchange the total system that is fixed, so that all ancillary systems (like a battery) need to be explicitly modelled and the system's Hilbert space does not change. More generally, once the free set $F_{AB}$ resulting from fusing two resource theories is known, one can also investigate allowed resource transformations in this new theory. In other words, on can ask what final states can be achieved via $\F_{AB}$ when the initial state is outside of $F_{AB}$.


\subsection{Elementary qubit example}
\label{sec:thermo_qubit}

We start with the simplest case of a two-level system that will serve us as an example to illustrate problems at hand. The two thermal states are then simply given by
\begin{equation}
    \v{\Gamma}=(\Gamma,1-\Gamma),\qquad \v{\gamma}=(\gamma,1-\gamma),
\end{equation}
with $\gamma \geq \Gamma$. In what follows, we will denote a thermal state by $\v{g}=(g,1-g)$, without specifying whether $\v{g}=\v\gamma$ or $\v{g}=\v\Gamma$. The set of all states achievable from a given state $\v{p}$ via thermal operations (i.e., the set of states $\v{q}$ such that $\v{p}\succ_{\v{g}} \v{q}$) is convex and its extremal points have been characterised in Ref.~\cite{lostaglio2018elementary} by studying the thermomajorisation order (see Lemma 12 therein). For $d=2$, there are just two extremal states achievable from~$\v{p}$: the state $\v{p}$ itself and
\begin{equation}
    \v{q} = \Pi_{\v{g}} \v{p},\qquad \Pi_{\v{g}}=\begin{pmatrix}
    1-\frac{1-g}{g} &1\\
    \frac{1-g}{g}& 0
    \end{pmatrix}.
\end{equation}

Importantly, applying two such non-trivial extremal transformations $\Pi_{\v{g}}$ in a row with respect to the same $\v{g}$ does not produce an extremal point. It is thus clear that the extremal states achievable from~$\v{p}$ by two agents after $N=2m+k$ strokes will have one of the following two forms
\begin{equation}
    \label{eq:qubit_extremal}
    \Pi_{\v{\Gamma}}^k(\Pi_{\v{\gamma}} \Pi_{\v{\Gamma}})^m \v{p},\qquad \Pi_{\v{\gamma}}^k(\Pi_{\v{\Gamma}} \Pi_{\v{\gamma}})^m \v{p},
\end{equation}
with $k\in \{0,1\}$ and $m\in\mathbb{N}$. We then note that the fixed points of the composition of extremal maps are given by (see Fig.~\hyperref[fig:qubit:thermo]{\ref{fig:qubit:thermo}a}):
\begin{subequations}
    \begin{align}
        (\Pi_{\v{\gamma}} \Pi_{\v{\Gamma}}) \tilde{\v{\gamma}}&=\tilde{\v{\gamma}}=(\tilde{\gamma},1-\tilde{\gamma}), \qquad \tilde{\gamma} = \frac{(2\gamma-1)\Gamma}{\Gamma+\gamma-1},\label{eq:gamma_tilde}\\
        (\Pi_{\v{\Gamma}} \Pi_{\v{\gamma}}) \tilde{\v{\Gamma}}&=\tilde{\v{\Gamma}}=(\tilde{\Gamma},1-\tilde{\Gamma}), \qquad {\tilde{\Gamma}} = \frac{(2\Gamma-1)\gamma}{\Gamma+\gamma-1}.\label{eq:Gamma_tilde}
    \end{align}
\end{subequations}
Note that
\begin{equation}
    \Pi_{\v{\Gamma}} \tilde{\v{\gamma}} = \tilde{\v{\Gamma}},\qquad \Pi_{\v{\gamma}} \tilde{\v{\Gamma}} = \tilde{\v{\gamma}},
\end{equation}
so that the fixed points can be mapped between each other by setting $k=1$ in Eq.~\eqref{eq:qubit_extremal}. 
Moreover, it is a~straightforward calculation to show that as $m$ grows, these fixed points are approached exponentially in $m$ inside the probability simplex, always from the side of the initial point as illustrated in Figs.~\hyperref[fig:qubit:thermo]{\ref{fig:qubit:thermo}b-\ref{fig:qubit:thermo}c}. More precisely, for every \mbox{$\v{p}=(p,1-p)$} and $m$ we have
\begin{equation}
    (\Pi_{\v{\gamma}} \Pi_{\v{\Gamma}})^m \v{p} = (q,1-q),\quad
    (\Pi_{\v{\Gamma}} \Pi_{\v{\gamma}})^m \v{p} = (r,1-r)
\end{equation}
with
\begin{subequations}
\begin{align}
    q&=\left(\frac{(1-\gamma)(1-\Gamma)}{\gamma\Gamma}\right)^m(p-\tilde{\gamma})+\tilde{\gamma},\\
    r&=\left(\frac{(1-\gamma)(1-\Gamma)}{\gamma\Gamma}\right)^m(p-\tilde{\Gamma})+\tilde{\Gamma},
\end{align}
\end{subequations}
which allows one to completely characterise the set $F_{AB}^{(N)}$ of states achievable after $N$ strokes. 

\begin{figure}[t]
\centering
\includegraphics[width=\columnwidth]{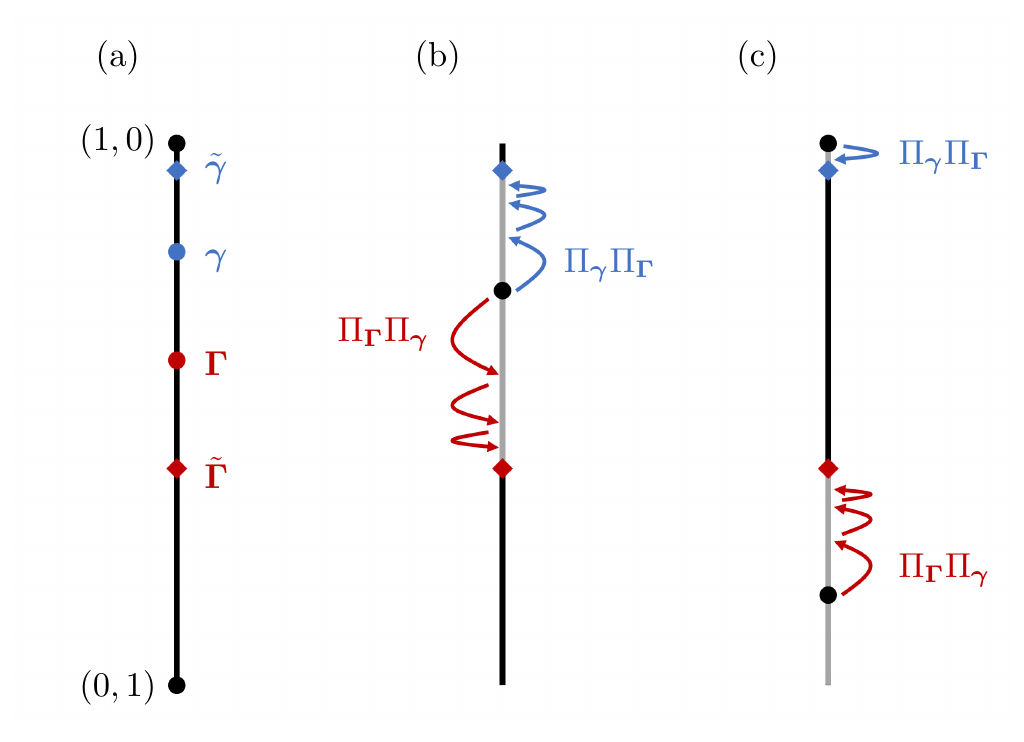}
\caption{\label{fig:qubit:thermo}\textbf{Two-level athermality engine.} (a) The space of incoherent states of a two-level system is given by a one-dimensional simplex with extremal points $(1,0)$ and $(0,1)$, corresponding to the ground and excited states, respectively. Thermal states with respect to the cold and hot temperatures are indicated by $\v{\gamma}$ and $\v{\Gamma}$, whereas fixed points of transformations $\Pi_{\v{\gamma}}\Pi_{\v{\Gamma}}$ and $\Pi_{\v{\Gamma}}\Pi_{\v{\gamma}}$ are denoted by $\tilde{\v{\gamma}}$ and $\tilde{\v{\Gamma}}$. Crucially, $\tilde{\v{\gamma}}$ ($\tilde{\v{\Gamma}}$) is always closer to the ground (excited) state than ${\v{\gamma}}$ (${\v{\Gamma}}$) is. (b) For initial states lying in the simplex between the fixed points $\tilde{\v{\gamma}}$ and $\tilde{\v{\Gamma}}$, the transformations $(\Pi_{\v{\gamma}}\Pi_{\v{\Gamma}})^m$ and $(\Pi_{\v{\Gamma}}\Pi_{\v{\gamma}})^m$ bring the initial state towards these fixed points ``from the inside''. (c). For initial states lying in the simplex outside the segment connecting the fixed points $\tilde{\v{\gamma}}$ and $\tilde{\v{\Gamma}}$, the transformations $(\Pi_{\v{\gamma}}\Pi_{\v{\Gamma}})^m$ and $(\Pi_{\v{\Gamma}}\Pi_{\v{\gamma}})^m$ bring the initial state towards these fixed points ``from the outside''.}
\end{figure}

Given all of the above, it is then clear that when the initial state $\v{p}$ belongs to 
\begin{equation}
    \label{eq:qubit:thermo:free}
    F_{AB}=\{\v{f}:~ \v{f}= (1-\lambda) \tilde{\v{\gamma}} +\lambda\tilde{\v{\Gamma}}~\mathrm{with}~0\leq\lambda\leq 1\},
\end{equation}
then the set of achievable states after arbitrarily large number of strokes  is given by $F_{AB}$. On the other hand, if $\v{p}\notin F_{AB}$ then there are two cases. If $\v{p}$ is closer to $\tilde{\v{\gamma}}$ in total variation distance than it is to $\tilde{\v{\Gamma}}$, then the set of achievable states is given by
\begin{equation}
    \label{eq:qubit:thermo:resource1}
    \!\! R_{\v{\Gamma}}(\v{p})=\{\v{r}:~ \v{r}= (1-\lambda) \v{p} + \lambda\Pi_{\v{\Gamma}}\v{p}~\mathrm{with}~0\leq\lambda\leq 1\},
\end{equation}
Otherwise, it is given by
\begin{equation}
    \label{eq:qubit:thermo:resource2}
    \!\! R_{\v{\gamma}}(\v{p})=\{\v{r}:~ \v{r}= (1-\lambda) \v{p} + \lambda\Pi_{\v{\gamma}}\v{p}~\mathrm{with}~0\leq\lambda\leq 1\}.
\end{equation}

Since both sets of original free states, $F_A=\{\v{\gamma}\}$ and $F_B=\{\v{\Gamma}\}$, are subsets of $F_{AB}$, we see that starting from either $\v{\gamma}$ or $\v{\Gamma}$ one can achieve any state belonging to $F_{AB}$. We thus conclude that $F_{AB}$ can be interpreted as the free set of the new resource theory that arises from fusing two resource theories of quantum thermodynamics with different temperatures. Moreover, any state $\v{p}$ outside $F_{AB}$ can be considered as a resource and its transformations are governed by Eqs.~\eqref{eq:qubit:thermo:resource1}-\eqref{eq:qubit:thermo:resource2}, which can be used to rigorously order states according to their resourcefulness. Finally, we want to emphasise that unless $\Gamma=1/2$ (which corresponds to the hot bath being at infinite temperature), the set $F_{AB}$ is not a full probability simplex. This means that fusing two thermodynamic resource theories results in a new non-trivial resource theory. 


\subsection{Bounding the set of achievable qudit states}
\label{sec:thermo:bounding}

Beyond the simple qubit case, the problem of fully characterising the set of achievable states $F_{AB}$ quickly becomes intractable due to the scaling of the number of conditions one needs to verify. More precisely, during one stroke, a given initial state $\v{p}$ can be generally transformed to $d!-1$ other  extremal states~\cite{lostaglio2018elementary}. Each of these states can then itself be an initial state of a system during the subsequent stroke, and so we see that, as the number $N$ of strokes grows, the number of states one needs to consider explodes as $(d!-1)^N$. Numerically, one can reduce the number of extremal states that need to be tracked by taking their convex hull after each stroke, and keeping only the extremal points thereof. This allows us to numerically consider large $N$ for $d\in\{3,4\}$, as we later discuss. 

However, we would also like to get some analytic insight into the set of states $F_{AB}$ achievable with the considered resource engine. Thus, in this section we first prove that there is an upper bound on $F_{AB}$, i.e., no matter how many strokes $N$ we allow for, there exist states that cannot be produced by a resource engine. Then, we provide an analytically simple construction of the lower bound of $F_{AB}$, i.e., we characterise a polytope of states that forms a subset of $F_{AB}$. Finally, we prove that for one bath with finite temperature and the other in the limit of infinite temperature, $\alpha>0$ and $\beta=0$, the full set of states becomes achievable, i.e., $F_{AB}$ is the whole probability simplex. 


\subsubsection{Upper bound}
\label{sec:thermo_upper}

In this section we will generalise the qubit result, showing that the set of free states $F_{AB}$ for a resource theory arising from fusing two resource theories of thermodynamics with different temperatures is generally a proper subset of the probability simplex. To achieve this, we will prove the following slightly more general theorem that constrains the possible final populations of the highest excited state under transformations $\F_{AB}$. 

\begin{thm}
    \label{thm:thermo_upper}
    Given an initial incoherent state $\v{p}$ and a~set of operations $\F_{AB}$ arising from fusing two resource theories of thermodynamics with different temperatures, all achievable final states $\v{q}$ satisfy
    \begin{equation}
        q_d \leq M(\v{p}):=\max \{p_d, \Gamma_d/\Gamma_{d-1}\},
    \end{equation}
    with $\v{\Gamma}$ denoting the thermal state of the system with respect to the higher temperature. Thus, $M(\v{p})$ is a monotone of the resulting resource theory.
\end{thm}

\begin{proof}

We start by introducing the following notation. For a set of $d$-dimensional incoherent states $X$ and a~thermal state $\v{\gamma}$, we denote by $\T_{\v{\gamma}}(X)$ a convex hull of all states that can be obtained from the elements of $X$ via thermal operations with a thermal state $\v{\gamma}$, i.e., 
\begin{equation}
    \T_{\v{\gamma}}(X):= \conv{\{\v{q}~|~\exists~\v{p}\in X:~\v{p}\succ_{\v{\gamma}}\v{q}\}}.
\end{equation}
We will use analogous notation for the thermal state given by $\v{\Gamma}$. Moreover, for any incoherent state $\v{p}$, we denote by $\bar{\v{p}}$ its version transformed in the following way:
\begin{equation}
    \bar{\v{p}} = (0,\dots,0,1-p_d,p_d)
\end{equation}
The crucial thing about this transformation is that
\begin{equation}
    \label{eq:thermo_app1}
    \v{p}\succ_{\v{\gamma}} \v{q} \Rightarrow \bar{\v{p}}\succ_{\v{\gamma}} \v{q} \quad \mathrm{and}\quad\v{p}\succ_{\v{\Gamma}} \v{q} \Rightarrow \bar{\v{p}}\succ_{\v{\Gamma}} \v{q},
\end{equation}
which is a simple consequence of the thermomajorisation order (see Appendix~\ref{app:thm1} for details).

With the introduced notation, it is then clear that:
\begin{subequations}
    \begin{align}
    \label{eq:cone_simplification1}
    \T_{\v{\gamma}}(\T_{\v{\Gamma}}(\v{p})) \subset \T_{\v{\gamma}}(\T_{\v{\Gamma}}(\bar{\v{p}}))&=\T_{\v{\gamma}}\left(\left\{ \v{q}~|~ \v{q}\in \T_{\v{\Gamma}}(\bar{\v{p}}) \right\}\right)\\
    &\subset \T_{\v{\gamma}}\left(\left\{ \bar{\v{q}}~|~ \v{q}\in \T_{\v{\Gamma}}(\bar{\v{p}}) \right\}\right),\label{eq:cone_simplification2}
\end{align}
\end{subequations}
where, for brevity, single-element sets $\{\v{p}\}$ and $\{\bar{\v{p}}\}$ are denoted by $\v{p}$ and $\bar{\v{p}}$, respectively. Using the definition of thermomajorisation, one can also show for $d\geq 3$ that (see Appendix~\ref{app:thm1} for details):
\begin{subequations}
    \begin{align}
    \!\!\!\!\max_{\v{q}}\{q_d|\v{q}\in \T_{\v{g}}(\bar{\v{r}})\}\!&=\!\left\{\!\begin{array}{cc}
         \!\! r_d& \mathrm{for~} r_d\geq \frac{g_d}{g_{d-1}+g_d},\!  \\
        \!\!\frac{g_d(1-r_d)}{g_{d-1}}& \mathrm{for~} r_d\leq \frac{g_d}{g_{d-1}+g_d},\!
    \end{array}\right.\label{eq:thermo_maxmin1}\\
    \!\!\!\!\min_{\v{q}}\{q_d|\v{q}\in\T_{\v{g}}(\bar{\v{r}})\}\!&=\!0,\label{eq:thermo_maxmin2}
\end{align}
\end{subequations}
for $\v{g}$ given by either $\v{\gamma}$ or $\v{\Gamma}$.

Using Eqs.~\eqref{eq:cone_simplification1}-\eqref{eq:cone_simplification2} and \eqref{eq:thermo_maxmin1}-\eqref{eq:thermo_maxmin2}, we can now upper bound the final population $q_d$ of the highest energy level after arbitrary one sequential interaction with two baths:
\begin{subequations}
    \begin{align}
    &\!\!\!\!\! \max_{\v{q}}\{q_d~|~\v{q}\in\T_{\v{\gamma}}(\T_{\v{\Gamma}}(\v{p})) \}\\
    &\leq    \max_{\v{q}}\{q_d~|~\v{q}\in\T_{\v{\gamma}}(\{\bar{\v{r}}~|~\v{r}\in\T_{\v{\Gamma}}(\bar{\v{p}})\})\}\\
    &=\max_{\v{q}} \max_{\v{r}\in\T_{\v{\Gamma}}(\bar{\v{p}})}\{q_d~|~\v{q}\in\T_{\v{\gamma}}(\bar{\v{r}})\}\\
    &= \max_{\v{r}\in\T_{\v{\Gamma}}(\bar{\v{p}})}\max_{\v{q}} \{ q_d~|~\v{q}\in\T_{\v{\gamma}}(\bar{\v{r}})\}\\
    &=\max_{\v{r}\in\T_{\v{\Gamma}}(\bar{\v{p}})} \left\{\begin{array}{cc}
         r_d& \mathrm{for~} r_d\geq \frac{\gamma_d}{\gamma_{d-1}+\gamma_d}  \\
        \frac{\gamma_d(1-r_d)}{\gamma_{d-1}}& \mathrm{for~} r_d\leq \frac{\gamma_d}{\gamma_{d-1}+\gamma_d}
    \end{array}\right.\\
    &\leq \max \left\{ \max_{\v{r}\in\T_{\v{\Gamma}}(\bar{\v{p}})} r_d, \frac{\gamma_d}{\gamma_{d-1}}(1-\min_{\v{r}\in\T_{\v{\Gamma}}(\bar{\v{p}})} r_d)  \right\}\\
    &\leq \max \left\{p_d, (1-p_d)\frac{\Gamma_d}{\Gamma_{d-1}},\frac{\gamma_d}{\gamma_{d-1}}\right\}\\
    &\leq \max \left\{p_d,\frac{\Gamma_d}{\Gamma_{d-1}},\frac{\gamma_d}{\gamma_{d-1}}\right\}=\max \left\{p_d,\frac{\Gamma_d}{\Gamma_{d-1}}\right\}.
\end{align}
\end{subequations}
Since after one sequential interaction the final population of the highest energy level is bounded by either the initial population of this level or by a constant, we conclude that the same bound holds after arbitrarily many repetitions (strokes).
\end{proof}
Using the above theorem it is then straightforward to prove the following.
\begin{cor}[Upper bound on $F_{AB}$]
    \label{cor:thermo_upper}
    The set $F_{AB}$ of free states arising from fusing two resource theories of thermodynamics with two different temperatures is bounded by
    \begin{equation}
        \forall \v{p}\in F_{AB}:\quad p_d \leq \frac{\Gamma_d}{\Gamma_{d-1}},
    \end{equation}
    where $\v{\Gamma}$ denotes the thermal state corresponding to higher temperature. 
\end{cor} 
As a final remark note that, as in the two-level case, when $\v{\Gamma}$ corresponds to a thermal state in the infinite temperature limit, the above bound trivialises. As we will show later, in this limit all states belong to $F_{AB}$.


\subsubsection{Lower bound}
\label{sec:thermo_lower}

In this section we first present a simple analytical construction of a subset of the full set of achievable states $F_{AB}$ for a resource theory arising from fusing two resource theories of thermodynamics with different temperatures. It is based on an iterative application of the modified qubit construction we saw in Sec.~\ref{sec:thermo_qubit}. We then also prove that, when the temperatures of the cold and hot bath satisfy certain conditions, $F_{AB}$ contains non-full-rank states, i.e., the athermality engine can reduce the rank of the initial full-rank thermal state. In particular, we will show when this rank reduction can be maximal, by proving when the ground state belongs to~$F_{AB}$.

We start by introducing a $d$-dimensional generalisations of the states $\tilde{\v{\gamma}}$ and $\tilde{\v{\Gamma}}$ from Eqs.~\eqref{eq:gamma_tilde}-\eqref{eq:Gamma_tilde}. We define them by setting one component as follows, 
\begin{subequations}
    \begin{align}
        \label{def:gamma_tilda_1}
        \tilde{\gamma}_1 = \frac{\Gamma_1(\gamma_1-\gamma_d)}{\Gamma_1(1-\gamma_d)-\Gamma_d(1-\gamma_1)},\\
        \tilde{\Gamma}_d = \frac{\Gamma_d(\gamma_1-\gamma_d)}{\Gamma_1(1-\gamma_d)-\Gamma_d(1-\gamma_1)},
    \end{align}
\end{subequations}
and leaving the remaining components proportional to thermal distribution:
\begin{equation}
    \tilde{\gamma}_k=\frac{1-\tilde{\gamma}_1}{1-\gamma_1}\gamma_k,\quad    \tilde{\Gamma}_k=\frac{1-\tilde{\Gamma}_d}{1-\Gamma_d}\Gamma_k.
\end{equation}
We then have the following simple lemma, the proof of which can be found in Appendix~\ref{app:lem3}.

\begin{lem}
    \label{lem:thermo_ext}
    For $d$-dimensional probability distributions $\v{p}$ and $\v{q}$ satisfying 
    \begin{align}
    \Gamma_d\leq p_d\leq \tilde{\Gamma}_d,\qquad
    \gamma_1\leq q_1\leq \tilde{\gamma}_1,
    \end{align}
    there exist $\v{p}',\v{p}'',\v{q}',\v{q}''$ such that
    \begin{align}    \label{eq:thermomajo_chain}
        \v{p}\succ_{\v{\gamma}} \v{p}' \succ_{\v{\Gamma}}\v{p}'',\qquad
        \v{q}\succ_{\v{\Gamma}} \v{q}' \succ_{\v{\gamma}}\v{q}'',
    \end{align}
    and
    \begin{subequations}
    \begin{align}
        \label{eq:pd_double_prime}
        p_d'' &= \frac{(1-\gamma_1)\Gamma_d}{(1-\gamma_d)\Gamma_1}p_d +\frac{(\gamma_1-\gamma_d)\Gamma_d}{(1-\gamma_d)\Gamma_1},\\
        \label{eq:q1_double_prime}
        q_1'' &= \frac{(1-\gamma_1)\Gamma_d}{(1-\gamma_d)\Gamma_1}q_1 +\frac{(\gamma_1-\gamma_d)}{(1-\gamma_d)},
    \end{align}
    \end{subequations}
    so that $p_d\leq p_d''\leq \tilde{\Gamma}_d$ and $q_1\leq q_1''\leq \tilde{\gamma}_1$. 
\end{lem}    

We can now employ the above lemma to prove that the set of achievable states $F_{AB}$ contains both $\tilde{\v{\gamma}}$ and $\tilde{\v{\Gamma}}$.

\begin{prop}
    \label{prop:tilde_achievable}
    The set $F_{AB}$ of free states arising from fusing two resource theories of thermodynamics with two different temperatures contains the states $\tilde{\v{\gamma}}$ and $\tilde{\v{\Gamma}}$. Moreover, the convergence to these states with the number of strokes $N$ is exponential.
\end{prop}
\begin{proof}
    We will only prove $\tilde{\v{\Gamma}}\in F_{AB}$, as the proof of \mbox{$\tilde{\v{\gamma}}\in F_{AB}$} proceeds analogously. The proof is simply based on applying Lemma~\ref{lem:thermo_ext} iteratively $m$ times.  Starting with the hot thermal state $\v{\Gamma}$ (preparation of which uses one stroke), we then apply the lemma $m$ times, which uses $2m$ strokes (since each application uses one connection with the cold and one with the hot bath). By a simple application of the recurrence formula, after $N=2m+1$ strokes, one generates a state $\v{q}$ with
    \begin{align}
        q_d&=\left(\frac{(1-\gamma_1)\Gamma_d}{(1-\gamma_d)\Gamma_1}\right)^m (\Gamma_d-\tilde{\Gamma}_d)+\tilde{\Gamma}_d,
    \end{align}
    and the remaining components can be thermalised with the hot bath, so that for $k\neq d$:
    \begin{equation}
        q_k=\frac{1-q_d}{1-\Gamma_d}\Gamma_k.
    \end{equation}
    Clearly, as $N\rightarrow\infty$, we get that $\v{q}\rightarrow \tilde{\v{\Gamma}}$, with the convergence being exponential in $m$, so also with $N$.
\end{proof}

Finally, we can use the above proposition iteratively to produce a whole polytope of states belonging to $F_{AB}$.

\begin{prop}[Lower bound on $F_{AB}$]
\label{prop:thermo_lower}
    Within a simplex of $d$-dimensional probability distributions, consider a polytope $P_d$ with $2^{d-1}$ extreme points $\{\v{f}^{\v{b}}\}$, each characterised by a bit string \mbox{$\v{b}=[b_2,\dots,b_d]$} of length $(d-1)$ via the following: 
    \begin{subequations}
    \begin{align}
      \!\!  f^{\v{b}}_d&=g_d^{(b_d)},\\
      \!\!  f^{\v{b}}_k&=\left(1-\!\!\sum_{l=k+1}^{d}f^{\v{b}}_l\right)g_k^{(b_k)}\quad\! \mathrm{for~} k\in\{d-1,\dots,2\},\\
      \!\!  f^{\v{b}}_1&=\left(1-\sum_{l=2}^{d}f^{\v{b}}_l\right),   
    \end{align}
    \end{subequations}
    where
    \begin{subequations}
    \begin{align}
        g_{k}^{(0)}:= \frac{\gamma_{k}(\Gamma_1-\Gamma_{k})}{\Gamma_1(1-\gamma_{k})-\Gamma_{k}(1-\gamma_1)},\\      
        g_{k}^{(1)}:= \frac{\Gamma_{k}(\gamma_1-\gamma_{k})}{\Gamma_1(1-\gamma_{k})-\Gamma_{k}(1-\gamma_1
        )}.
    \end{align}
\end{subequations}
Then, the set $F_{AB}$ of free states arising from fusing two resource theories of thermodynamics with two different temperatures contains $P_d$.
\end{prop}
\begin{proof}
  
    Since the set $F_{AB}$ is convex, we only need to show that $F_{AB}$ contains $\v{f}^{\v{b}}$ for any choice of the bit string $\v{b}$. Let us then fix $\v{b}$ and note that $g_d^{(0)}=\tilde{\gamma}_d$ and $g_d^{(1)}=\tilde{\Gamma}_d$. Thus, depending on $b_d$, we start by generating $\tilde{\v{\Gamma}}$ or $\tilde{\v{\gamma}}$ (by Proposition~\ref{prop:tilde_achievable} we can do this, as both of these states belong to $F_{AB}$). As a result, the occupation of level $d$ is given by ${f}_d^{\v{b}}$. Next, set $k=1$ and treat the first $(d-k)$ levels as an unnormalised $(d-k)$-level state. Depending on $b_{d-k}$, again using Proposition~\ref{prop:tilde_achievable} we generate versions of the states $\tilde{\v{\Gamma}}$ or $\tilde{\v{\gamma}}$ restricted to these $(d-k)$ levels (and unnormalised). Repeating this for $k\in\{2,\dots,d-2\}$, we obtain the state with the occupation of the level $k$ given by $f_k^{\v{b}}$. Finally, the occupation of the ground state is given by $f_1^{\v{b}}$ simply by the normalisation condition. 
\end{proof}

In Fig.~\ref{fig:3-4_thermo}, we present how this lower bound (and the upper bound from the previous section) compares with the full set of achievable states $F_{AB}$ obtained numerically for two examples of three- and four-level athermality engines. As can be seen, in the first case, the numerically obtained $F_{AB}$ contains only full rank states; whereas in the second case, up to numerical precision, $F_{AB}$ seems to contain the rank-1 ground state. Since the pairs of cold and hot temperatures in the two studied examples were different, one may wonder whether the ability of the athermality engine to reduce rank depends on the temperature difference. In the next section we will see that, irrespective of the cold bath, when the hot bath is at infinite temperature, $F_{AB}$ contains all states (so also the ones with the reduced rank). Here, we show that even if the temperature of the hot bath is not infinite, but it is large enough and the cold bath is cold enough, $F_{AB}$ contains a rank-1 ground state (see Appendix~\ref{app:ground} for the proof of the following proposition).

\begin{prop}
    \label{prop:ground}
    Assume $\alpha$ to be large enough and $\beta$ to be small enough, so that
    \begin{align}
        \label{eq:coldenough} 
        &\gamma_1 > \frac{1}{2}\quad\text{and}\quad  \Gamma_1 < \Gamma_d+\Gamma_{d-1}.
    \end{align}
    Then, the set $F_{AB}$ of free states arising from fusing two resource theories of thermodynamics with two different temperatures contains the ground state. Achieving this state requires $N\to\infty$ strokes.
\end{prop}

\begin{figure}
    \centering
    \includegraphics[width=\columnwidth]{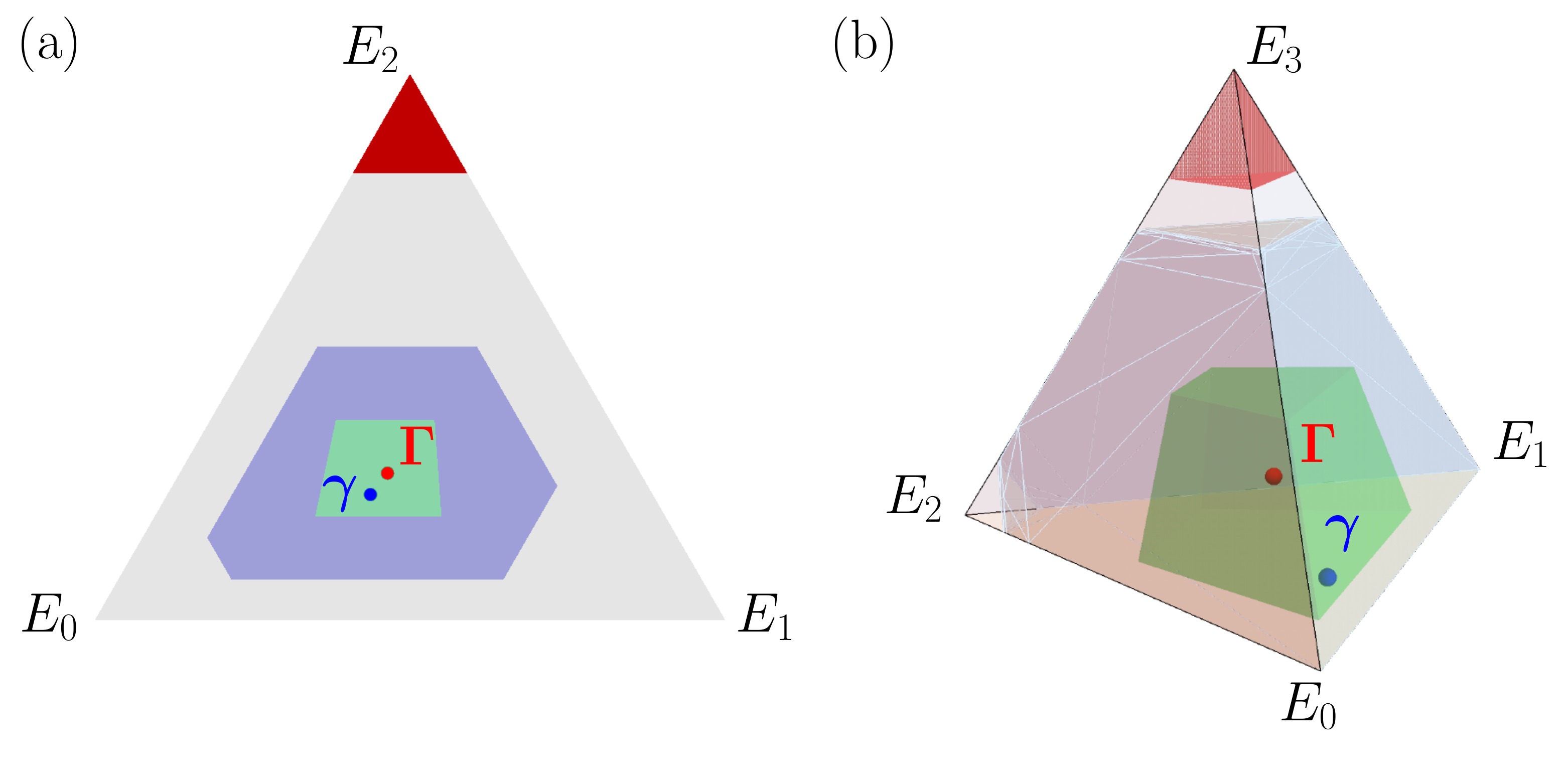}
    \caption{\textbf{Three- and four-level athermality engines.} The space of incoherent states of a $d$-level system is given by a~$(d-1)$-dimensional probability simplex [(a) $d=3$ and (b) $d=4$]. The set $F_{AB}$ of states achievable by an athermality engine after $N\to\infty$ strokes is presented in blue (region obtained numerically for large enough $N$, so that increasing it does not lead to visible changes), the region excluded from $F_{AB}$ by Corollary~\ref{cor:thermo_upper} is depicted in red, and the lower bound obtained in Proposition~\ref{prop:thermo_lower} is indicated in green.  Parameters chosen: (a)~$\alpha=1/3$, $\beta=1/5$, $E_k=k$; (b)~$\alpha=1$, $\beta=1/4$, $E_k=k$. }
    \label{fig:3-4_thermo}
\end{figure}


\subsubsection[Full set of achievable states for infinite temperature]{\texorpdfstring{Full set of achievable states for $\beta=0$}{Full set of achievable states for infinite temperature}}
\label{sec:thermo_infinite}

The upper bound constraining the set $F_{AB}$ from Corollary~\ref{cor:thermo_upper} becomes trivial when $\beta=0$, i.e., when the hot bath is at infinite temperature. One can then wonder, whether this bound is not tight, or rather in this special case there are no constraints and the set $F_{AB}$ coincides with the full probability simplex. The following theorem shows that the latter is the case.
\begin{thm}
    The set $F_{AB}$ of free states arising from fusing two resource theories of thermodynamics, one with finite and the other with infinite temperature, is given by the full probability simplex. Moreover, the convergence to the full simplex with the number of strokes $N$ is exponential.
\end{thm}
\begin{proof}
    By Proposition~\ref{prop:tilde_achievable}, $\tilde{\v{\gamma}}\in F_{AB}$ and the convergence to this state is exponential in the number of strokes. However, in the limit of infinite temperature $\Gamma_k=1/d$ for all $k\in\{1,\ldots,d\}$, and so, by Eq. \eqref{def:gamma_tilda_1},
    \begin{equation}
        \tilde{\v{\gamma}}=(1,0,\dots,0).
    \end{equation}
    Finally, since
    \begin{equation}
        (1,0,\dots,0)\succ_{\v{\Gamma}} \v{p}
    \end{equation}
    for every $\v{p}$, the set $F_{AB}$ contains all probability distributions. 
\end{proof}

We thus see that the set $F_{AB}$ of achievable states trivialises when the hot temperature is infinite ($\beta=0$), even without cold temperature being zero. One can see this even without analysing the formal proof. Indeed, it suffices to observe the following. Firstly, in the case of $\beta=0$, the state $\tilde{\v{\gamma}}$ given by Eq.~\eqref{def:gamma_tilda_1}, is the ground state $(1,0,0,\ldots,0)$. Secondly, $\tilde{\v{\gamma}}$ belongs to $F_{AB}$ due to Proposition~\ref{prop:tilde_achievable}. Thirdly, for $\beta=0$ the set $\mathcal{F}_B$ contains permutations, and so Bob can transform the ground state to the highest excited state. Finally, since this state, independently of $\alpha$ and $\beta$, thermomajorises all other states, one can transform it to any state within the probability simplex.


\section{Coherence engine}
\label{sec:coherence}


\subsection{Setting the scene}
\label{sec:coherence_setting}

We now depart from athermality engines and switch to investigating engines fueled by the resource of quantum coherence. The theoretical framework established to quantify the amount of coherence present in a state of the system is known under the collective name of resource theories of coherence~\cite{aberg2006superposition,baumgratz2014quantifying}. Their aim is to quantitatively capture the departure from principles of classical physics due to the quantum superposition principle or, in other words, to quantify the ability of a system to undergo quantum interference effects. As with every resource theory, they are defined by identifying the set of free states and free operations. While all such theories agree on the form of the free states (they are simply given by states diagonal in the distinguished basis), there are plenty of choices for the set of free operations (see Ref.~\cite{streltsov2017colloquium} for a review on this subject). To avoid making any specific choices (and also for simplicity), here we will restrict our considerations to a reversible subtheory, where all different resource theories of coherence coincide. Namely, we will only consider unitary transformations, and so the set of free operations will be given by unitaries diagonal in the distinguished basis\footnote{Note that, via the Stinespring dilation, this can serve as a starting point for analysing free operations given by more general quantum channels.}. To make things even simpler, in our toy theory of coherence engines we will restrict ourselves to pure states. Thus, the only free states of the theory will be given by the distinguished basis states.

More formally, we consider two agents, $A$ and $B$, that can prepare the following sets of free states
\begin{subequations}
    \begin{align}
        F_A&=\{|i\rangle\}_{i=1}^d,\\
        F_B&=\{U^{\dagger}|i\rangle\}_{i=1}^d,
    \end{align}
\end{subequations}
where $U\in \mathcal{U}_d(\mathbb{C})$ is a fixed unitary of order $d\geq 2$ over the field $\mathbb{C}$ describing the relative orientation of the two bases. To define the sets of free operations, let  
\begin{equation}
   \!\! \mathcal{DU}_d(\mathbb{C})=\left\{\sum_{i=1}^d u_{ii}\ketbra{i}{i}:\,|u_{ii}|=1\right\}
\end{equation}
denote the subgroup of $\mathcal{U}_d(\mathbb{C})$ consisting of unitary matrices diagonal in the computational basis. 
Now, the sets of free operations are given by
\begin{subequations}
\begin{align}
\label{def:A}
\F_A&=\mathcal{DU}_d(\mathbb{C}),\\
\label{def:B}
\F_B&=\{U^{\dagger}DU:\;\;D\in \mathcal{DU}_d(\mathbb{C})\}.
\end{align}
\end{subequations}

Note that this setup is directly connected to the problem of Hamiltonian control~\cite{ramakrishna}. More precisely, one can imagine having control over a quantum system with an intrinsic Hamiltonian $H_1$, and being able to change it at will to $H_2$ by some external influence. By carefully choosing the timings and lengths of such interventions, one can then perform a sequence of unitaries $e^{iH_1 t_1}$, $e^{iH_2 t_2}$, $\dots$. Under the assumption of incommensurable spectra, this is then equivalent to a sequence of arbitrary unitaries diagonal in two different bases (of eigenstates of $H_1$ and $H_2$), which is the action of a coherence engine. To give a simple example, consider a spin-1/2 system, initially prepared in an up state along the $\hat{z}$ axis, and two agents, $A$ and $B$, who can turn on and off the magnetic field along the $\hat{z}$ and $\hat{n}$ axis, respectively. Such control of the magnetic field means that they can rotate the state of the system on the Bloch sphere around these axes. The aim is to have full control of the system's state, i.e., to have a protocol that maps the initial state to any state $\ket{\psi}$, optimally with as few steps as possible. The introduced coherence engine formalism allows to address exactly this problem and to answer whether such a full control is possible and, if so, what is the minimal number of steps needed.

The concept of coherence engines also bears resemblance to the  problem of generating universal gate sets in quantum computing~\cite{lloyd1995almost,weaver2000universality}. However, there are crucial differences that distinguish the two concepts. Universal gate sets in quantum computing refer to finite sets of quantum gates that, when combined in various ways, can be used to approximate any unitary operation on a~quantum computer. These gate sets are typically fixed and predetermined, e.g., all single-qubit gates along with a two-qubit entangling gate (like a CNOT gate). The ability to approximate any unitary operation using these gates is a fundamental requirement for the construction of a quantum computer. On the other hand, coherence engines, while sharing some similarities in terms of their ability to implement general unitary transformations, offer a more flexible and continuous approach. Specifically, two parties can implement any diagonal unitary in two different bases, thus enabling the realization of a continuous set of operations. In the case when the bases are appropriately chosen, the two parties can also approximate any unitary, just like in the case with a universal gate set. The continuous spectrum of operations may give more flexibility, but the number of uses/strokes strongly depends on the ``angle'' between Alice's and Bob's bases. Therefore, it is difficult to determine whether the universal gate set or a coherence engine will give a smaller number of usages to obtain a desired approximation of a~general unitary.

Similarly as in the case of athermality engines, fusing two  resource theories described by $(F_A,\F_A)$ and $(F_B,\F_B)$, by means of a coherence engine, leads to a new resource theory with free states \mbox{$F_{AB}\supseteq F_A \cup F_B$} and free operations \mbox{$\F_{AB}\supseteq \F_A \cup \F_B$}. The main questions that we address here are as follows. First, in Sec.~\ref{sec:coherence_operations}, we ask when $\F_{AB}=\mathcal{U}_d(\mathbb{C})$, i.e., under what conditions a coherence engine can generate any unitary operation. Next, in Sec.~\ref{sec:coherence_bounds}, we address this question quantitatively by asking about the number of strokes needed to achieve any operation from $\mathcal{U}_d(\mathbb{C})$. We then switch the attention from operations to states and investigate the task of reaching a given pure quantum state via a number of strokes of a resource engine. In particular, in Sec.~\ref{sec:coherence_optimal}, we analyse how quickly one can reach a state that is maximally resourceful (here: mutually coherent) with respect to both Alice's and Bob's constraints. However, before we present all of the above, we start with the simplest qubit example in Sec.~\ref{sec:coherence_qubit} to illustrate problems at hand.


\subsection{Elementary qubit example}
\label{sec:coherence_qubit}

For the qubit case, it will be convenient to use the Bloch sphere representation of a pure state $\ket{\psi}$:
\begin{equation}
    \ketbra{\psi}{\psi}=\frac{\iden+\hat{\v{r}}\cdot\v{\sigma}}{2},
\end{equation}
where $\hat{\v{r}}=(r_x,r_y,r_z)$ is the Bloch vector satisfying \mbox{$r_x^2+r_y^2+r_z^2=1$}, and $\v{\sigma}=(\sigma_x,\sigma_y,\sigma_z)$ is a vector of Pauli matrices. Let us also recall that an arbitrary single qubit unitary operator can be written in the form
\begin{align}
    R_{\hat{\v{n}}}(\theta)=e^{i \theta\hat{\v{n}}\cdot \v{\sigma} },
\end{align}
and its action on the Bloch vector $\hat{\v{r}}$ is to rotate it by an angle $\theta$ around an axis specified by a three-dimensional unit vector $\hat{\v{n}}$.

Without loss of generality, we can then choose Alice's distinguished basis to be given by the eigenstates of $\sigma_z$, and Bob's to be given by the eigenstates of $\hat{\v{n}}\cdot\v{\sigma}$ with \mbox{$\hat{\v{n}}=(\sin(\alpha),0,\cos(\alpha))$}. Thus, the set $F_A$ consists of states corresponding to the north and south pole of the Bloch sphere, whereas $\F_A$ is given by rotations $R_{\hat{z}}(\theta)$ around the $z$ axis by an arbitrary angle $\theta$. Similarly, the set $F_B$ consists of states from $F_A$ rotated in the $xz$ plane of the Bloch sphere by an angle $\alpha$, whereas $\F_B$ is given by rotations $R_{\hat{\v{n}}}(\theta)$ around the axis $\hat{\v{n}}$, i.e., around the axis tilted with respect to the ${z}$ axis by an angle~$\alpha$ (see Fig.~\ref{fig:new_rot}).

Using the Euler angles decomposition, we see that a~rotation around any axis $\hat{\v{m}}$ by any angle $\theta$ can be written as
\begin{align}
\label{eq:Euler_angles}
R_{\hat{\v{m}}}(\theta)=R_{\hat{z}}\left(\beta\right)R_{\hat{z}_{\perp}}\left(\gamma\right)R_{\hat{z}}\left(\delta\right),
\end{align}
where $\hat{z}_{\perp}$ denotes a unit vector perpendicular to $\hat{z}$, and an appropriate choice of $\beta,\delta\in[0,2\pi]$ and $\gamma\in[0,\pi]$ has to be made. Therefore, if $\hat{\v{n}}=\hat{z}_{\perp}$ (i.e., $\alpha=\pi/2$), Alice and Bob only need to perform three operations alternately (three strokes of a resource engine are needed) to generate arbitrary unitary operation. In other words, if the two distinguished bases are mutually unbiased (and so coherences with respect to them form complementary resources), then any unitary transformation can be generated with three strokes of a coherence engine.

Whenever $\hat{\v{n}}\neq\hat{z}_{\perp}$ (i.e., $\alpha<\pi/2$), the number of strokes increases to $\lceil\pi/\alpha\rceil+1$. This fact, although studied in a different context, has been established, e.g., in Refs.~\cite{lowenthal, lowenthal2}; cf. also Ref.~\cite{hamada}. Indeed, the first and the last operation in Eq.~\eqref{eq:Euler_angles} remain the same, but $R_{\hat{z}_{\perp}}\left(\gamma\right)$ has to be replaced by 
a product of alternated rotations around vectors $\hat{z}$ and $\hat{\v{n}}$, respectively. In order to justify that it is possible, as well as to estimate the number of necessary alternated rotations, let us again refer to 
Eq.~\eqref{eq:Euler_angles}, and observe that for any $\tilde{\theta}\in\mathbb{R}$, there exist some $\tilde{\beta},\tilde{\delta}\in[0,2\pi]$ and $\tilde{\gamma}\in[0,\pi]$, depending on both $\tilde{\theta}$ and $\alpha$, such that
\begin{align}\label{eq:Euler_angles_j}
R_{\hat{\v{n}}}\left(\tilde{\theta}\right)=
R_{\hat{z}}\left(\tilde{\beta}\right)R_{\hat{z}_{\perp}}\left(\tilde{\gamma}\right)R_{\hat{z}}\left(\tilde{\delta}\right).
\end{align}
As a consequence, we get
\begin{align}\label{eq:Euler_cor}
R_{\hat{z}_{\perp}}\left(\tilde{\gamma}\right)=
R_{\hat{z}}\left(-\tilde{\beta}\right)R_{\hat{\v{n}}}\left(\tilde{\theta}\right)R_{\hat{z}}\left(-\tilde{\delta}\right).
\end{align}

\begin{figure}[!t]
    \centering
    \includegraphics[width=\columnwidth]{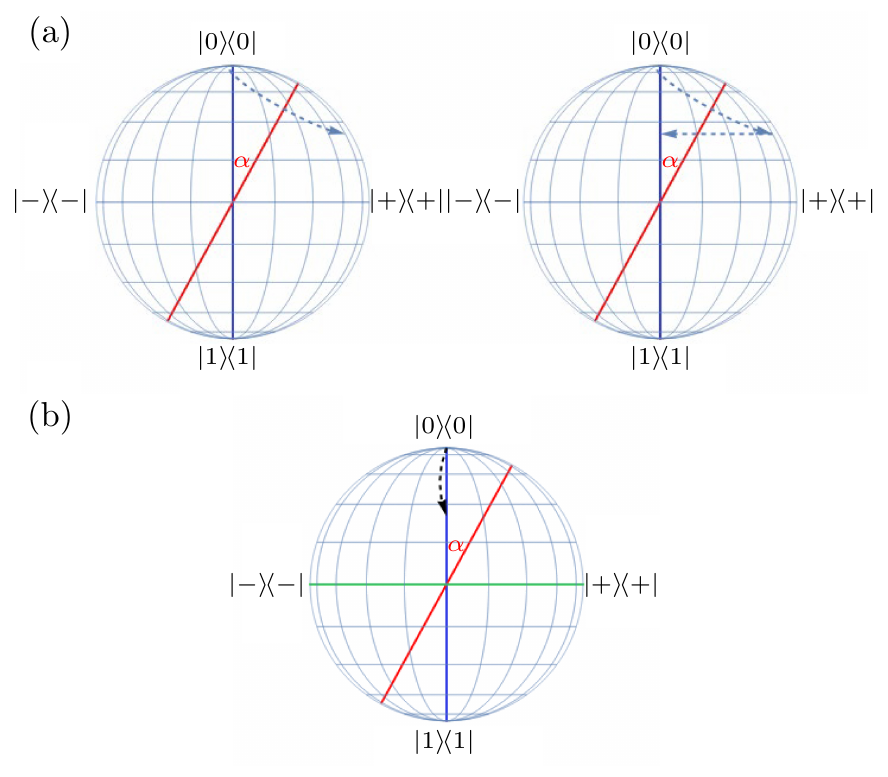};
	\caption{\textbf{Two-level coherence engine.} Rotating alternately in the $\sigma_z$ and the $\hat{\v{n}}\cdot\v{\sigma}$ eigenbases (i.e., around the $\hat{z}$ and $\hat{\v{n}}$ axes, respectively) can generate a single rotation around arbitrary axis, in particular around $\hat{y}$. (a) Illustration of the right hand side of Eq.~\eqref{eq:Euler_cor}. (b) Illustration of the left hand side of Eq.~\eqref{eq:Euler_cor}. 
    }
\label{fig:new_rot}
\end{figure}

However, it is important to observe that Eq.~\eqref{eq:Euler_cor} only holds for these values $\tilde{\gamma}$ which satisfy Eq.~\eqref{eq:Euler_angles_j} with certain $\tilde{\beta},\tilde{\delta},\tilde{\theta}$. Intuitively, the smaller the angle $\alpha$ between the vectors $\hat{z}$ and $\hat{\v{n}}$, the smaller the angle $\tilde{\gamma}$, by which we can rotate around the vector $\hat{z}_{\perp}$ (see Fig.~\ref{fig:new_rot}). To be more precise, $\tilde{\gamma}$ cannot be greater than $2\alpha$. This indicates that to generate $R_{\hat{z}_{\perp}}(\gamma)$ for an arbitrarily big $\gamma$, we need to perform $k$ rotations around $\hat{z}_{\perp}$ by angles not bigger than $2\alpha$ in a row, namely
\begin{align}    R_{\hat{z}_{\perp}}\left(\gamma\right)=R_{\hat{z}_{\perp}}\left(\gamma_1\right)\ldots R_{\hat{z}_{\perp}}\left(\gamma_k\right).
\end{align}
This, in turn, implies that $k=\lceil \gamma/(2\alpha)\rceil$. Since we aim to estimate the number of necessary rotations to generate any unitary matrix, we have to consider the worst case scenario, which is $k=\lceil\pi/(2\alpha)\rceil$. 
Note that, after applying Eq.~\eqref{eq:Euler_cor} $k$ times (with $\gamma_i,\beta_i,\theta_i,\delta_i$ in the places of $\tilde{\gamma},\tilde{\beta},\tilde{\theta},\tilde{\delta}$, respectively), we end up with
\begin{align}
\begin{aligned}
R_{\hat{z}_{\perp}}\left(\gamma\right)=&
R_{\hat{z}}\left(-{\beta}_1\right)R_{\hat{\v{n}}}\left(\theta_1\right)R_{\hat{z}}\left(-{\delta}_1-\beta_2\right)\ldots \\
&\times R_{\hat{z}}\left(-\delta_{k-1}-{\beta}_k\right)R_{\hat{\v{n}}}\left(\theta_k\right)R_{\hat{z}}\left(-{\delta}_k\right),
\end{aligned}
\end{align}
which, due to Eq.~\eqref{eq:Euler_angles}, finally implies
\begin{align}
\begin{aligned}
R_{\hat{\v{m}}}\left(\theta\right)=&
R_{\hat{z}}\left(\beta-{\beta}_1\right)R_{\hat{\v{n}}}\left(\theta_1\right)R_{\hat{z}}\left(-{\delta}_1-\beta_2\right)\ldots \\
&\times R_{\hat{z}}\left(-\delta_{k-1}-{\beta}_k\right)R_{\hat{\v{n}}}\left(\theta_k\right)R_{\hat{z}}\left(-{\delta}_k+\delta\right).
\end{aligned}
\end{align}
Therefore, the total number of alternated operations from $\F_A$ and $\F_B$ needed to generate an arbitrary unitary operation is equal to $2k+1=2\lceil \pi/(2\alpha)\rceil+1\geq \lceil\pi/\alpha\rceil+1$. For details of this reasoning, see, e.g., Ref.~\cite{hamada}. An even more precise analysis, approaching the problem from the perspective of Lie algebras, is presented in Refs.~\cite{lowenthal}~and~\cite{lowenthal2}, and the result obtained there is precisely $\lceil\pi/\alpha\rceil+1$. We conclude that the incompatibility of the two bases measured by $\alpha$ quantifies the power of a coherence engine.

The analysis to estimate the number of strokes of a~resource engine needed to reach any state when starting from a~free state (to focus attention -- an element of~$F_A$) is very similar. The only crucial difference is that we can now choose to start from this pole of the Bloch sphere that is closer to the final state, and we do not need the initial rotation around $\hat{z}$ (since states from $F_A$ are anyway invariant under such transformations). As a~consequence, it is enough to select $\gamma$ in Eq.~\eqref{eq:Euler_angles} from $[0,\pi/2]$ (instead of $[0,\pi]$), which requires $\left\lceil {\pi}/{2\alpha}\right\rceil$ strokes. This allows Alice and Bob to create a state with arbitrary polar angle, and the final rotation around $\hat{z}$ allows for the choice of arbitrary azimuthal angle, thus getting to all states on the Bloch sphere. In conclusion, being restricted to the sets $F_A$, $F_B$ of their free states and $\mathcal{F}_A$, $\mathcal{F}_B$ of their free operations, Alice and Bob can reach any pure quantum state with $\left\lceil {\pi}/{2\alpha}\right\rceil+1$ strokes of a resource engine.

Instead of asking how many strokes it takes to get any state (or operation), we can also try to determine what is achievable with just three strokes, i.e., one round of communication between Alice and Bob. In particular, we will be interested in the conditions under which they can generate an optimal state (for both Alice and Bob) in this way. Of course, we first need to define what we mean by optimal in this case. It is natural to see the optimal state as the one that is maximally resourceful for both parties, and as such we can choose a maximally mutually coherent state with respect to Alice's and Bob's bases (see Appendix~\ref{app:mutually_coh} or Refs.~\cite{korzekwa_jennings_rudolph,Idel_wolf,puchala_rudnicki}). For two qubit bases given by eigenstates of $\hat{\v{m}}\cdot\v{\sigma}$ and $\hat{\v{n}}\cdot\v{\sigma}$, a maximally mutually coherent state is given by a state with a Bloch vector $\hat{\v{r}}\propto\hat{\v{m}}\times\hat{\v{n}}$, where $\times$ denotes the cross product. Thus, in our case with $\hat{\v{m}}=\hat{z}$ and $\hat{\v{n}}=(\sin(\alpha),0,\cos(\alpha))$, it is an eigenstate of $\sigma_y$, which corresponds to a state described by a Bloch vector with the polar angle $\pi/2$ and the azimuthal angle $\pi/2$. Now, note that the first stroke just prepares one of the two states from $F_A$, and the third stroke is an operation from $\F_A$ that cannot change the polar angle of a state. We thus see that in order to create the optimal state in three strokes, an operation from~$\F_B$ must transform the polar angle of the state from $F_A$ (with polar angles 0 or $\pi$) to $\pi/2$. Since a rotation around an axis tilted by $\alpha$ with respect to the initial Bloch vector can at most create an angle $2\alpha$ between the initial and final vectors, we see that this is possible if and only if \mbox{$\alpha\in\left[\frac{\pi}{4},\frac{3\pi}{4}\right]$}. Therefore, as long as the largest absolute value of the overlap between the distinguished basis states is smaller or equal than $\cos(\pi/8)$, a coherence engine can produce a maximally mutually resourceful state in just three strokes. For a~step-by-step analytical proof of the presented argument, see Appendix~\ref{app:mutually_coh}.


\subsection{Condition on getting all operations}
\label{sec:coherence_operations}

For the elementary case of a qubit system, we have seen that, as long as the distinguished bases of Alice and Bob do not coincide, a coherence engine can generate an arbitrary unitary operation. For higher-dimensional systems, however, this is no longer the case. Thus, the aim of this section is to formulate conditions under which Alice and Bob, restricted to performing operations from $\F_A$ and $\F_B$, respectively, are able to generate an arbitrary unitary. In other words, we ask under what conditions $\F_{AB}$ becomes the full set of unitary operations.

For a given $U\in\mathcal{U}_d(\mathbb{C})$, let us introduce a matrix \hbox{$P_U=(p_{ij})_{i,j=1}^d$} over the field $\mathbb{R}_+$, which corresponds to $U$ in the following sense:
\begin{align}\label{def:P_U}
\begin{aligned}
p_{ij}=\left\{\begin{array}{ll}
0&\text{for }u_{ij}=0,\\
1&\text{for }u_{ij}\neq 0.
\end{array}\right.
\end{aligned}
\end{align}
Next, let us formulate the following two conditions on a~matrix \hbox{$U\in \mathcal{U}_d(\mathbb{C})$}:
\begin{itemize}
\item[(H1)] \phantomsection \label{cnd:H1}
There exist a constant $M\in\mathbb{N}$ and matrices $D_1,\ldots,D_{2M}\in\mathcal{DU}_d(\mathbb{C})$ such that 
\begin{align}
D_1U^{\dagger}D_2UD_3U^{\dagger}D_4U\ldots D_{2M-1}U^{\dagger}D_{2M}U
\end{align}
is a matrix with all non-zero entries.
\item[(H2)] \phantomsection \label{cnd:H2}
There exists a constant $M\in\mathbb{N}$ such that 
\begin{align}
\left(P_U^TP_U\right)^M
\end{align}
is a matrix with all non-zero entries.
\end{itemize}
Obviously, hypothesis \hyperref[cnd:H2]{(H2)} is much easier to verify in practise than hypothesis \hyperref[cnd:H1]{(H1)}. However, in Appendix~\ref{app:equivalence}, we prove the following.
\begin{prop}\label{prop:equiv_H1_H2}
    Hypotheses \hyperref[cnd:H1]{(H1)} and \hyperref[cnd:H2]{(H2)} are equivalent.
\end{prop}

We will also need the following lemma, which follows from~\cite[Lemmas 2 and 3]{borevich}, and has been initially made in the proof of~\cite[Proposition 7]{schmid}.
\begin{lem}\label{lemma}
A subgroup $\mathcal{V}$ in $\mathcal{U}_d(\mathbb{C})$ containing $\mathcal{DU}_d(\mathbb{C})$ and a matrix \hbox{$W=(w_{ij})_{i,j=1}^d$}, such that $w_{ij}\neq 0$ for any $i,j\in\{1,\ldots,d\}$, is a full unitary group.
\end{lem}

We can now state and prove the main result of this section.

\begin{thm}
\label{proposition}
If $U\in\mathcal{U}_d(\mathbb{C})$, appearing in the definition of $\F_B$, satisfies \hyperref[cnd:H2]{(H2)}, then any unitary matrix can be written as a~product comprised of unitary matrices from $\F_A$ and from $\F_B$. 
\end{thm}
\begin{proof}
Let $M\in\mathbb{N}$ and {$D_1,\ldots,D_{2M}\in \mathcal{DU}_d(\mathbb{C})$} be such that  \hyperref[cnd:H2]{(H2)} holds. Then, using Proposition~\ref{prop:equiv_H1_H2}, the matrix
\begin{align}\label{non-zero}
D_1U^{\dagger}D_2UD_3U^{\dagger}D_4U\ldots D_{2M-1}U^{\dagger}D_{2M}U
\end{align}
has all non-zero entries. Note that this matrix is a product comprised of $M$ matrices from $\F_A$ and $M$ matrices from $\F_B$. Next, using Lemma~\ref{lemma}, we obtain that any unitary matrix can be written as a product of a number of matrices from Eq.~\eqref{non-zero} and a number of appropriately chosen diagonal unitary matrices. \end{proof}

\begin{rem}
According to~\cite[Lemmas 2 and 3]{borevich}, we do not even need to demand in Lemma~\ref{lemma} that $\mathcal{V}$ contains a matrix with all non-zero entries. Instead we can only assume that $\mathcal{V}$ contains a matrix $Q=(q_{ij})_{i,j=1}^d$ such that
\begin{align}
\forall_{k,l\in\{1,\ldots,d\}}\;\exists_{m\in\{1,\ldots,d\}}\;\;\;q_{km}\neq0\;\;\;\text{and}\;\;\;q_{lm}\neq 0.
\end{align}
This would, however, yield (due to~\cite[Lemmas 2 and 3]{borevich}) that some matrix $W$ with all non-zero entries also belongs to $\mathcal{V}$.  
\end{rem}

To see that there exist both unitary matrices for which condition~\hyperref[cnd:H2]{(H2)} is satisfied and unitary matrices for which it is violated, let us refer to the theory of finite Markov chains (see Appendix~\ref{app:markov} for a brief review of the subject). For a chosen finite Markov chain, let $T$ denote its transition matrix. It is well known that in the case of $T$ being irreducible and aperiodic~\cite{haggstrom_2002}, e.g.,
\begin{align}\label{def:P}
{T}=0.1\left[\begin{array}{cccccc}
5&2&2&1&0&0\\
2&5&2&1&0&0\\
1&1&1&1&2&4\\
2&2&1&1&1&3\\
0&0&3&4&1&2\\
0&0&1&2&6&1
\end{array}\right],
\end{align}
there exists a finite constant $M$ such that for every \mbox{$m\geq M$} the $m$-th step transition matrix $T^m$ consists only of non-zero elements (cf. Proposition~\ref{prop:_irreducible_aperiodic} in Appendix~\ref{app:markov}).

Keeping this in mind, one can try to find a unitary matrix $U\in\mathcal{U}_d(\mathbb{C})$ such that $P_U^TP_U$ (with $P_U$ induced by $U$ via Eq.~\eqref{def:P_U}) has the same pattern of zero and non-zero elements as some irreducible and aperiodic transition matrix $T$, e.g., in the case of $T$ given by Eq.~\eqref{def:P} one can think of
\begin{align}\label{def:U}
\begin{aligned}
U=\frac{1}{2} \left[\begin{array}{cccccc}
0&0&-1&1&1&1\\
1&1&-1&-1&0&0\\
-\sqrt{2}&\sqrt{2}&0&0&0&0\\
-1&-1&-1&-1&0&0\\
0&0&0&0&\sqrt{2}&-\sqrt{2}\\
0&0&1&-1&1&1
\end{array}\right].
\end{aligned}
\end{align}
We then immediately get that $U$ enjoys \hyperref[cnd:H2]{(H2)} (for $U$ given by Eq.~\eqref{def:U}, condition \hyperref[cnd:H2]{(H2)} holds with $M=2$). Another example, illustrating how to find in practice the number $M$ appearing in \hyperref[cnd:H2]{(H2)} (and not only claim that such a~number exists and is finite), is provided in Appendix~\ref{app:markov}.

Proceeding now to ``negative'' examples, one shall take into account that both assumptions in Proposition~\ref{prop:_irreducible_aperiodic} in Appendix~\ref{app:markov} concerning finite Markov chains (i.e. irreducibility and aperiodicity of a transition matrix $T$) are essential. Hence, we may conclude the following:
\begin{enumerate}
    \item\label{1} Any unitary matrix $U$ such that $P_U^TP_U$ is reducible violates condition~\hyperref[cnd:H2]{(H2)}.
    
    \item\label{2} Any unitary matrix $U$ such that $P_U^TP_U$ is periodic violates condition~\hyperref[cnd:H2]{(H2)}. Note, however, that for $U\in\mathcal{U}_d(\mathbb{C})$ every diagonal element of $P_U^TP_U$ is positive, whence $P_U^TP_U$ cannot be periodic unless it is reducible -- cf. Propositions~\ref{prop:loop}~and~\ref{prop:graph} in Appendix~\ref{app:markov} -- and thus we end up in case~\ref{1}.
\end{enumerate}

Now, it suffices to observe that all matrices $U\in\mathcal{U}_d(\mathbb{C})$, which become block-diagonal after being multiplied by an appropriate permutation matrix (possibly also the identity matrix), violate condition~\hyperref[cnd:H2]{(H2)} (we conjecture that these are the only matrices violating condition~\hyperref[cnd:H2]{(H2)}). As precise examples we can either provide block-diagonal unitary matrices, such as 
\begin{align}
\left[
\begin{array}{cccc}
1&0&0&0\\
0&1&0&0\\
0&0&\frac{1}{\sqrt{2}}&-\frac{1}{\sqrt{2}}\\
0&0&\frac{1}{\sqrt{2}}&\frac{1}{\sqrt{2}}
\end{array}
\right],\quad
\left[
\begin{array}{cccc}
-\frac{1}{\sqrt{2}}&0&0&\frac{1}{\sqrt{2}}\\
0&-\frac{1}{\sqrt{2}}&\frac{1}{\sqrt{2}}&0\\
0&\frac{1}{\sqrt{2}}&\frac{1}{\sqrt{2}}&0\\
\frac{1}{\sqrt{2}}&0&0&\frac{1}{\sqrt{2}}
\end{array}
\right]
\end{align}
or non-block-diagonal unitary matrices, such as
\begin{align}
\left[
\begin{array}{cccc}
0&1&0&0\\
0&0&1&0\\
0&0&0&1\\
1&0&0&0
\end{array}
\right],\quad
\left[
\begin{array}{cccc}
0&-\frac{1}{\sqrt{2}}&\frac{1}{\sqrt{2}}&0\\
-\frac{1}{\sqrt{2}}&0&0&\frac{1}{\sqrt{2}}\\
\frac{1}{\sqrt{2}}&0&0&\frac{1}{\sqrt{2}}\\
0&\frac{1}{\sqrt{2}}&\frac{1}{\sqrt{2}}&0
\end{array}
\right],
\end{align}
which can be transformed to block-diagonal matrices by multiplying them by appropriate permutation matrices.


\subsection{Bounds on the number of strokes to get all operations}
\label{sec:coherence_bounds}

We start with a statement on the lower bound on the number of resource engine's strokes needed to produce any unitary operation, the proof of which can be found in Appendix~\ref{app:proofs-lower-bound}. 

\begin{prop}
\label{prop:lower_bound}
Let $d\geq 3$ and suppose that the sets $\F_A$ and $\F_B$ are given by Eqs.~\eqref{def:A}-\eqref{def:B} with an arbitrarily fixed matrix $U\in\mathcal{U}_d(\mathbb{C})$. The number $N$ of resource engine's strokes (i.e., the number $2M$ of alternated Alice's and Bob's operations) needed to generate an arbitrary operation $V\in \mathcal{U}_d(\mathbb{C})$   
is bounded from below by 
\begin{align}
    N\geq \frac{2\log\left(d-1\right)}{\log\left((d-2)c_U+1\right)}, 
\end{align}
where
\begin{align}\label{def:c_U}
\begin{aligned}
c_U:&=\max_{a,b\in\{1,\ldots,d\}:\;a\neq b}\sum_{j=1}^d\left|\bar{u}_{j,a}\right|\left|u_{j,b}\right|\\
&=\max_{a,b\in\{1,\ldots,d\}:\;a\neq b}\left\langle a\left|\left(X_U^TX_U\right)\right|b\right\rangle
\end{aligned}
\end{align}
and $[X_U]_{i,j}=|u_{i,j}|$ for all \mbox{$i,j\in\{1,\ldots,d\}$}. 
\end{prop}
While the above lower bound is not tight (which can be seen by noting that for two mutually unbiased bases with $c_U=1$, the bound states $N\geq 2$, and we know that already for $d=2$ one needs 3 strokes), it nevertheless provides useful information when Alice's and Bob's bases do not differ much and so more strokes are required. This is illustrated in Fig.~\ref{fig:lower}, where we plot this bound for a family of unitaries $U$ defining $\F_B$ continuously connecting the identity matrix with the Fourier matrix.

\begin{figure}[t]
	\centering	\includegraphics[width=0.99\columnwidth]{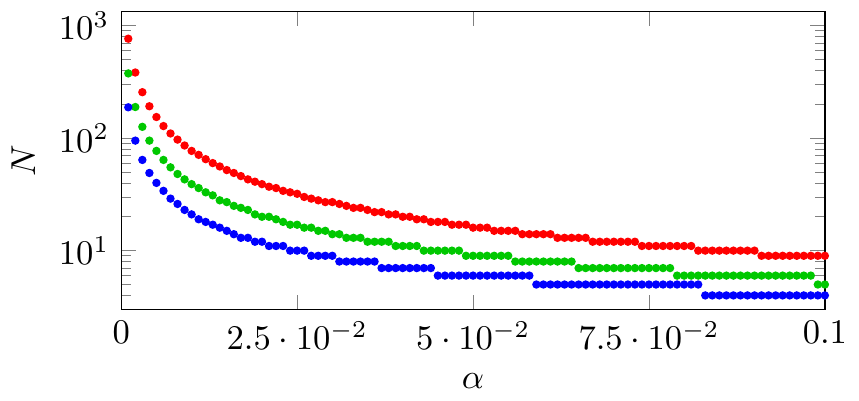};
	\caption{\textbf{Lower bound.} The lower bound from Proposition~\ref{prop:lower_bound} for the number of coherence engine's strokes needed to generate every unitary transformation, when the unitary $U$ defining $\F_B$ is given by $F_d^\alpha$, with $F_d$ denoting the $d$-dimensional Fourier matrix, and the three plots correspond to $d=3$ (top red), $d=5$ (middle green) and $d=10$ (bottom blue).
	}
	\label{fig:lower}
\end{figure}

The upper bound on the number of strokes needed to produce any operation $V\in\mathcal{U}_d(\mathbb{C})$ (with an arbitrary positive even number $d$) will be implicitly expressed by means of the number $N_F$ of strokes needed to generate the Fourier matrix.
\begin{prop}
    Let $d\geq 4$ be an arbitrary even number and suppose that the sets $\F_A$ and $\F_B$ are given by Eqs.~\eqref{def:A}-\eqref{def:B} with an arbitrarily fixed matrix \mbox{$U\in\mathcal{U}_d(\mathbb{C})$}. Moreover, let $N_F$ denote the number of resource engine's strokes needed to generate the Fourier matrix~$F_d$ of order $d$. Then, the number $N$ of resource engine's strokes needed to generate an arbitrary operation $V\in \mathcal{U}_d(\mathbb{C})$ is bounded from above by    
    \begin{align}
        N\leq 6d\left(N_F+1\right)+1.
    \end{align}
\end{prop}

\begin{proof}
    In Ref.~\cite{LopezPastor}, an analytical and deterministic procedure to design an implementation of an arbitrary unitary transformation is provided with the use of Fourier transforms and phase masks. According to the results established there, an arbitrary unitary of order $d$ requires, in the case of $d$ being an even number, $6d$ Fourier matrices and $6d+1$ diagonal unitary matrices (implemented by phase masks). The assertion of this proposition is then a~straightforward consequence of this fact.
\end{proof}

In the end let us indicate that, according to Ref.~\cite{schmid},  every $U\in\mathcal{U}_d(\mathbb{C})$ (no matter whether $d$ is even or odd)  can be decomposed as a sequence of unitary circulant and unitary diagonal matrices. What is more, it can further be written as a product comprised of unitary diagonal matrices and the discrete Fourier transforms (since any circulant matrix can be written as a product of the Fourier matrix, some diagonal matrix, and the inverse of a Fourier matrix). However, to the best of our knowledge, there are no other results, apart from Ref.~\cite{LopezPastor}, to which we refer in the proof above, on the number of unitary diagonal matrices and Fourier matrices needed to generate a given unitary matrix. In Ref.~\cite{product_circulant_diagonal}, the authors prove that an arbitrary complex matrix can be written as a product of circulant and diagonal (but not necessarily unitary) matrices with the number of factors being $2d-1$ at most.


\subsection{Condition on getting the optimal state}
\label{sec:coherence_optimal}

In this section, we focus on the problem of using a coherence engine to generate a maximally mutually coherent state $\ket{\psi_*}$ with respect to Alice's and Bob's distinguished bases. Such a state has the following form,
\begin{subequations}    
\begin{align}
\left|\psi_*\right\rangle
&=\frac{1}{\sqrt{d}}
\left(e^{i\alpha_1}|1\rangle+\ldots+e^{i\alpha_d}\left|d\right\rangle\right)\\
&=\frac{1}{\sqrt{d}}
\left(e^{i\beta_1}U^{\dagger}|1\rangle+\ldots+e^{i\beta_d}U^{\dagger}\left|d\right\rangle\right),
\end{align}
\end{subequations}
and we discuss their existence for arbitrary $d$-dimensional systems in Appendix~\ref{app:mutually_coh}. In particular, we ask for which matrices $U$ (which determine the set $\F_B$), the state $\ket{\psi_*}$ can be generated in only three strokes (i.e., Alice prepares a state from $F_A$, sends it to Bob who transforms it via one of the unitaries belonging to $\F_B$, and finally sends it back to Alice, who performs a unitary from $\F_A$). A~set of necessary conditions for this task is given by the following result.

\begin{prop}\label{rem:necessary}
Let $U=(u_{ij})_{i,j=1}^d\in\mathcal{U}_d(\mathbb{C})$. The following conditions are necessary for a coherence engine to produce a maximally mutually coherent state with just three strokes:
\begin{align}
\label{eq:necessary}
\underset{l}{\exists}\;\underset{m}{\forall}\;\max_{i}\left|\bar{u}_{im}u_{il}\right|&\leq\frac{1}{2}\left(\sum_{j=1}^d\left|\bar{u}_{jm}u_{jl}\right|+\frac{1}{\sqrt{d}}\right),
\end{align}
where $l,m,i\in\{1,\dots d\}$. Equivalently,
\begin{align}
\begin{aligned}\label{eq:necessary_equiv}
\!\!\! \max_{l}\min_{m}\bigg(\!\sum_{j=1}^d\frac{\left|\bar{u}_{jm}u_{jl}\right|}{2}
\!-\!\max_{i}\left|\bar{u}_{im}u_{il}\right|\!\bigg)\!\geq\! -\frac{1}{2\sqrt{d}}.
\end{aligned}
\end{align}
\end{prop}

Detailed proof of the above proposition can be found in Appendix~\ref{app:mutually_coh} and here we briefly explain the idea behind it. First, one notes that the task of generating a maximally mutually coherent state $\ket{\psi_*}$ (using just three strokes of a~coherence engine) is possible if there exist $l\in\{1,\ldots,d\}$ and $D\in\mathcal{DU}_d(\mathbb{C})$ such that for a~given $U=(u_{ij})_{i,j=1}^d\in\mathcal{U}_d(\mathbb{C})$ the matrix $U^{\dagger}DU$ has a~flat $l$-th column (cf. Eqs. \eqref{eq:flat} and \eqref{eq:phases} in Appendix~\ref{app:mutually_coh}). Note that this requirement can be rewritten as follows: there are $l\in\{1,\ldots,d\}$, $\xi_1,\ldots,\xi_d\in(0,2\pi]$ and $\kappa_1,\ldots,\kappa_d\in(0,2\pi]$ such that 
\begin{align}
    \underset{m\in\{1,\ldots,d\}}{\forall}\quad \sum_{j=1}^de^{i\xi_j}\bar{u}_{jm}u_{jl}=e^{i\kappa_m}\frac{1}{\sqrt{d}}.
\end{align} 
The generalised polygon inequalities then imply the assertion of Proposition~\ref{rem:necessary}.

The obvious consequence of Proposition~\ref{rem:necessary} is the following corollary.
\begin{cor}\label{corollary}
If $U=(u_{ij})_{i,j=1}^d\in\mathcal{U}_d(\mathbb{C})$, appearing in the definition of $\F_B$, is such that
\begin{align}\label{eq:violated_necessary_cnd}
\min_{i}\max_{j}\left|u_{ij}\right|^2>\sqrt{\frac{1+\frac{1}{\sqrt{d}}}{2}},
\end{align}
then one of the necessary conditions from Eq.~\eqref{eq:necessary} (namely the one with $m=l$) is violated, and thus Alice and Bob are not able to generate a~maximally mutually coherent state with just three strokes of a coherence engine.
\end{cor}

It is also interesting to observe that whenever a matrix $U\in\mathcal{U}_d(\mathbb{C})$ is, in a certain sense, too close to the some permutation matrix, then at least one of the necessary conditions given in Eq.~\eqref{eq:necessary} is violated. This is more formally captured by the following.
    
\begin{prop}\label{prop:close_to_permutation}
Let $U=(u_{ij})_{i,j=1}^d\in\mathcal{U}_d(\mathbb{C})$ be a matrix defining $\F_B$ via Eq.~\eqref{def:B}. If for some permutation matrix $\Pi\in\mathcal{U}_d(\mathbb{C})$ and some $D\in\mathcal{DU}_d(\mathbb{C})$, we have
\begin{align}
\label{eq:close_to_permutation}
\left\|U-D\Pi\right\|_{\text{HS}}^2< 2-2\sqrt{\frac{1+\frac{1}{\sqrt{d}}}{2}},
\end{align}
where $\|\cdot\|_{\text{HS}}$ denotes the Hilbert-Schmidt norm, 
then Alice and Bob are not be able to generate a maximally mutually coherent state with just two resource engine's strokes.
\end{prop}

\begin{proof}
First of all, note that
\begin{align}
\begin{aligned}
\left\|U-D\Pi\right\|_{\text{HS}}^2
&=\text{Tr}\left((U-D\Pi)(U-D\Pi)^{\dagger}\right)\\
&=\text{Tr}\left(UU^{\dagger}+D\Pi(D\Pi)^{\dagger}-2\text{Re}\left(U\Pi^{\dagger}D\right)\right)\\
&=2d-2\text{Re}\left(\sum_{k=1}^d\langle k\left|U\Pi^{\dagger}D\right|k\rangle\right)\\
&=2d-2\sum_{k=1}^d\text{Re}\left(u_{k\Pi^{-1}(k)}\right)e^{i\theta_k},
\end{aligned}
\end{align}
which, for appropriately chosen diagonal matrix $D$ and appropriately chosen permutation matrix $\Pi$, is equal to
\begin{align}
2d-2\sum_{k=1}^d\left|u_{k\Pi^{-1}(k)}\right|
=2d-2\sum_{k=1}^d\max_{j\in\{1,\ldots,d\}}\left|u_{kj}\right|.
\end{align}
Assumption from Eq.~\eqref{eq:close_to_permutation} therefore implies
\begin{align}
d-\sum_{k=1}^d\max_{j\in\{1,\ldots,d\}}\left|u_{kj}\right|< 1-\sqrt{\frac{1+\frac{1}{\sqrt{d}}}{2}},
\end{align}
whence
\begin{align}
\sum_{k=1}^d\max_{j\in\{1,\ldots,d\}}\left|u_{kj}\right|> (d-1)+\sqrt{\frac{1+\frac{1}{\sqrt{d}}}{2}}.
\end{align}

Now, it suffices to observe the following general fact: for any $|x_1|,\ldots,|x_d|\in[0,1]$ such that 
\begin{align}\label{eq:technical}
\sum_{i=1}^d\left|x_i\right|>(d-1)+\delta\;\;\;\text{for some}\;\;\;\delta>0,
\end{align}
we have 
\begin{align}\label{eq:technical_}
    \left|x_i\right|>\delta\;\;\;\text{for all}\;\;\;i\in\{1,\ldots,d\}.
\end{align}
Indeed, let $|x_1|,\ldots,|x_d|\in[0,1]$ be such that they satisfy Eq.~\eqref{eq:technical}, and suppose (contrary to the above claim) that there exists $i_0\in\{1,\ldots,d\}$ such that $|x_{i_0}|\leq\delta$. Then
\begin{align}
\sum_{i=1}^d\left|x_i\right|\leq \delta+\sum_{i\in\{1,\ldots,d\}\backslash\{i_0\}}\left|x_i\right|\leq (d-1)+\delta,
\end{align}
which contradicts the undertaken assumption in Eq.~\eqref{eq:technical}.

Applying the above to $\{|u_{kj}|\}_{j=1}^d$, which enjoys Eq.~\eqref{eq:technical} with $\delta=\sqrt{(1+1/\sqrt{d})/2}$, we obtain, due to Eq.~\eqref{eq:technical_}, the following:
\begin{align}
\max_{j\in \{1,\ldots,d\}}\left|u_{kj}\right|>\sqrt{\frac{1+\frac{1}{\sqrt{d}}}{2}}\;\;\;\text{for any}\;\;\;k\in\{1,\ldots,d\}.
\end{align}
Equivalently,
\begin{align}
\min_{k\in\{1,\ldots,d\}}\max_{j\in\{1,\ldots,d\}}\left|u_{kj}\right|>\sqrt{\frac{1+\frac{1}{\sqrt{d}}}{2}}.
\end{align}
According to Corollary~\ref{corollary}, we therefore obtain the assertion of this proposition.

\end{proof}


\section{Outlook}
\label{sec:outlook}

In this work we have introduced the notion of a resource engine and used it to analyse the possibility and complexity of generating quantum states and channels from two different sets of free states and operations. The main motivation behind introducing and investigating the resource engine was due to the fact that it provides a natural way of fusing two (or more) resource theories, in the spirit of recent works on multi-resource theories~\cite{sparaciari2020first}. In a sense, it allows one to study how compatible various constraints on allowed transformations are. What we have done here is just to start this kind of research by introducing ideas and analysing simple toy examples, but we hope that a formal mathematical framework allowing for fusing arbitrary resource theories can be developed. Two potential ways to achieve this could be to extend to multiple resources the framework of general convex resource theories~\cite{takagi2019general}, or the very recent framework for quantifying resource-dependent complexity of quantum channels~\cite{araiza2023resource}.

We also want to point out that the resource engine perspective may provide a unified framework to study seemingly unrelated problems within the field of quantum information. First, as explained in Sec.~\ref{sec:coherence_setting}, our toy example, given by a setting with two agents being constrained to performing unitaries diagonal in two different bases, can be directly related to the problem of compiling universal quantum circuits via Hamiltonian control~\cite{lloyd1995almost,weaver2000universality,ramakrishna}. Thus, the results on resource engines could be applied to optimise quantum control and circuit compilation. Moreover, as discussed in the paper, a similar connection can be made with the problem of performing arbitrary optical linear transformations~\cite{LopezPastor,pereira2020minimum} (e.g., Alice can be restricted to performing only phase masks, while Bob can only perform Fourier transforms or beam splitters).

Another potential use of resource engines is in the area of quantum error correction. Here, transversal gates are crucial for fault-tolerant quantum computation because of their robustness to noise, as well as their simplicity. The problem is that no quantum error correcting code can transversely implement a universal gate set~\cite{eastin2009restrictions}. However, by combining two error correcting codes with complementary transversal gate sets (so that together they are universal) via a method known as code switching~\cite{anderson2014fault,jochym2014using}, one can achieve fault-tolerance. Therefore, studies of resource engines with two restricted sets corresponding to two transversal gate sets, can allow for optimising fault-tolerant quantum computation based on code switching.

Finally, one can also find applications for resource engines in analysing the generation and distribution of entanglement in quantum networks~\cite{acin2007entanglement,kimble2008quantum,perseguers2013distribution}. Such a network is represented by a graph with vertices corresponding to communicating parties and edges denoting connections between parties via quantum channels. Thus, generating a specific multipartite entangled state in a network can be decomposed into steps, where at each step two parties connected via a channel can process their part of the state. The simplest scenario of a linear network with three nodes (so $A$ connected with $B$ and $B$ connected with $C$), can then be investigated using our resource engine with two constraints given by local operations on $AB$ and local operations on $BC$. Extending the resource engine to more constraints (equivalent of a heat engine having access to baths at various temperatures) would allow onto deal with more complex quantum networks and could help in analysing optimal protocols for generating particular entangled states (e.g., minimising the number of communication rounds).


\subsection*{Acknowledgements} 

HW-S and ZP acknowledge financial support from SONATA BIS programme of National Science Centre Poland (grant no. 2016/22/E/ST6/00062), while KK was supported by the Foundation for Polish Science through TEAM-NET project (contract no. POIR.04.04.00-00-17C1/18-00). 

\bibliographystyle{quantum}
\bibliography{Bib_channels}


\onecolumngrid
\appendix


\section{Thermomajorisation}
\label{app:thermo}

Consider a $d$-dimensional quantum system described by a Hamiltonian $H$ with eigenvalues \mbox{$E_1\leq\dots\leq E_d$}. Denote a fixed inverse temperature of the environment by~$\beta$. Then, the thermal Gibbs distribution of the system is given by
\begin{equation}
    \v{\gamma}=\frac{1}{Z}\left(e^{-\beta E_1},\dots,e^{-\beta E_d}\right),\qquad Z=\sum_{i=1}^d e^{-\beta E_i}.
\end{equation}
In order to construct a thermomajorisation curve of a probability vector $\v{p}$ (which describes the occupations of different energy levels in a given energy-incoherent state), one first needs to reorder the entries of $\v{p}$ according to the so-called $\beta$-order. To do so, denote by $\v{\pi}(\v{p})$ the reordering of $\{1,\dots, d\}$ that sorts $p_i/\gamma_i$ in a non-increasing order,
\begin{equation}
	\label{eq:betaorder}
	\frac{p_{\pi_i(\v{p})}}{\gamma_{\pi_i(\v{p})}} \geq \frac{p_{\pi_{i+1}(\v{p})}}{\gamma_{\pi_{i+1}(\v{p})}}.
\end{equation}
The $\beta$-ordered version of $\v{p}$ is then given by $\v{p}^\beta$ with $p^\beta_i=p_{\pi_i(\v{p})}$. Next, the thermomajorisation curve of $\v{p}$ relative to $\v{\gamma}$ is given by a piecewise linear concave curve on a plane that connects the points $\v{l}^{(j)}$ given by
\begin{align}
	\v{l}^{(j)} = \left(\sum_{i=1}^j \gamma_{\pi_i(\v{p})}, \sum_{i=1}^j p_{\pi_i(\v{p})}\right)
\end{align}
for $j\in\{0,\dots,d\}$, where $\v{l}^{(0)}:= (0,0)$. Finally, $\v{p}$ is said to thermomajorise $\v{q}$ (relative to $\v{\gamma}$), denoted \mbox{$\v{p} \succ_{\v{\gamma}} \v{q}$}, when the thermomajorisation curve of $\v{p}$ is never below that of $\v{q}$.


\section{Missing details in the proof of Theorem~\ref{thm:thermo_upper}}
\label{app:thm1}

We first need to show that Eq.~\eqref{eq:thermo_app1} holds, i.e.,
\begin{equation}
    \v{p}\succ_{\v{\gamma}} \v{q} \Rightarrow \bar{\v{p}}\succ_{\v{\gamma}} \v{q} \quad \mathrm{and}\quad\v{p}\succ_{\v{\Gamma}} \v{q} \Rightarrow \bar{\v{p}}\succ_{\v{\Gamma}} \v{q}.
\end{equation}
Without loss of generality, we can just prove the first of the above keeping $\v{\gamma}$ general. The way to prove it is to show that $\bar{\v{p}}\succ_{\v{\gamma}} \v{p}$ which, due to transitivity of thermomajorisation order, leads to $\bar{\v{p}}\succ_{\v{\gamma}} \v{q}$. For $d=2$ this is trivial, since $\bar{\v{p}}=\v{p}$. For $d\geq 3$, first consider two energy levels $E_i\leq E_j$, and a state $\v{r}$ with $r_i=0$ and $r_j> 0$. Then, there exists a~Gibbs-preserving operation that, on the two-dimensional subspace spanned by energy levels $E_i$ and $E_j$, achieves an arbitrary occupation $(a,r_j-a)$ for $a\in[0,r_j]$. It is explicitly given by a matrix 
\begin{equation}
\label{eq:Tij}
\renewcommand{\arraystretch}{2}
    T=\begin{bmatrix}
        1-\dfrac{\gamma_j a}{\gamma_i r_j}  &\dfrac{a}{r_j} \\ 
        \dfrac{\gamma_j a}{\gamma_i r_j} &1-\dfrac{a}{r_j} 
        \end{bmatrix}\oplus \mathbbm{1}_{{\backslash}(ij)},
\end{equation}
where $\mathbbm{1}_{{\backslash}(ij)}$ denotes the $(d -2) \times (d-2)$ identity matrix on the subspace of all energy levels except $E_i$ and $E_j$. In other words, it is always possible to move probability down the energy ladder to unoccupied levels. Since the first $(d-2)$ levels of $\bar{\v{p}}$ are unoccupied, a sequence of such operations can move the occupations from $\bar{p}_{d-1}$ down to all the energy levels, leading to a final distribution $\v{p}$.

Next, we proceed to demonstrating that Eqs.~\eqref{eq:thermo_maxmin1}-\eqref{eq:thermo_maxmin2} hold, i.e., that for $d\geq 3$ we have
\begin{subequations}
    \begin{align}\label{eq:qmax}
    \max_{\v{q}}\{q_d|\v{q}\in \T_{\v{g}}(\bar{\v{r}})\}&=\left\{\begin{array}{cc}
         r_d& \mathrm{for~} r_d\geq \frac{g_d}{g_{d-1}+g_d},  \\
        \frac{g_d(1-r_d)}{g_{d-1}}& \mathrm{for~} r_d\leq \frac{g_d}{g_{d-1}+g_d},
    \end{array}\right.\\
    \min_{\v{q}}\{q_d|\v{q}\in\T_{\v{g}}(\bar{\v{r}})\}&=0.
\end{align}
\end{subequations}
We start with the second statement. It can be simply proved by observing that for $d \geq 3$ we have (by definition) $\bar{r}_1=0$, and so using $T$ from Eq.~\eqref{eq:Tij} with $i=1$, $j=d$ and $a=r_d$, all the occupation of level $d$ can be moved to level~$1$. Since probabilities must be non-negative, the achieved value of $q_d=0$ is minimal.

To prove Eq.~\eqref{eq:qmax}, we need to resort to thermomajorisation curves. Since $\bar{\v{r}}$ has only two non-zero entries, its thermomajorisation curve is given by three points, the first one being $(0,0)$ and the third one being $(1,1)$. The middle point is given by $(g_d,r_d)$ if $r_d\geq g_d/(g_{d-1}+g_d)$ and by $(g_{d-1},1-r_d)$ otherwise. In the first case, for the thermomajorisation curve of $\v{q}$ to lie below that of $\v{r}$, its $y$ component must be no higher than $r_d$ at the $x$ component~$g_d$. Given that the thermomajorisation curve is concave, this means that $q_d\leq r_d$. In the second case, the $y$ component of the thermomajorisation curve of $\v{r}$ at $x$ component $g_d$ is given by $(1-r_d)g_d/g_{d-1}$. Therefore, for the thermomajorisation curve of $\v{q}$ to lie below that of $\v{r}$, the maximum value of $q_d$ must be upper bounded by $(1-r_d)g_d/g_{d-1}$. Moreover, this upper bound can be achieved by simply choosing $\v{q}$ with that value of $q_d$ and all others proportional to the entries of the Gibbs state, i.e., $q_i=g_i(1-q_d)/(1-g_d)$ for $i<d$.


\section{Proof of Lemma~\ref{lem:thermo_ext}}
\label{app:lem3}

\begin{proof}

We start from proving the first half of the lemma (related to distributions $\v{p}$, $\v{p}'$ and $\v{p}''$). Consider a state $\check{\v{p}}$ obtained from $\v{p}$ by thermalising the first $(d-1)$ levels of the system with respect to the cold bath, i.e.,
\begin{equation}
    \check{\v{p}}=\left(\frac{1-p_d}{1-\gamma_d}\gamma_1,\frac{1-p_d}{1-\gamma_d}\gamma_2,\dots, \frac{1-p_d}{1-\gamma_d}\gamma_{d-1},p_d\right).
\end{equation}
Clearly, $\v{p}\succ_{\v{\gamma}}\check{\v{p}}$, as thermalising a subset of states is a thermal operation. Since 
\begin{equation}
    \check{p}_d=p_d\geq \Gamma_d \geq \gamma_d,    
\end{equation}
and the remaining entries $\check{p}_i$ are proportional to the entries of the thermal distribution $\v{\gamma}$, the thermomajorisation curve of $\check{\v{p}}$ with respect to $\v{\gamma}$ is defined by three points: $(0,0)$, $(\gamma_d,p_d)$ and $(1,1)$. Thus, for the $x$ component $\gamma_1$, its $y$ component is given by
\begin{equation}
    p_1'=p_d+\frac{\gamma_1-\gamma_d}{1-\gamma_d}(1-p_d)=\frac{1-\gamma_1}{1-\gamma_d} p_d + \frac{\gamma_1-\gamma_d}{1-\gamma_d}.
\end{equation}
This means that $\check{\v{p}}$ can be transformed to $\v{p}'$ with the above occupation of the ground state $p_1'$ and the remaining occupations $p_i'$ simply chosen to be proportional to the entries of the thermal distribution $\v{\gamma}$. Finally, since thermomajorisation is transitive, we have \mbox{$\v{p}\succ_{\v{\gamma}} \v{p}'$}.

The next step is to start with $\v{p}'$ and thermalise levels $\{2,\dots,d\}$ with respect to the hot bath, resulting in
\begin{equation}
    \check{\v{p}}'=\left(p_1', \frac{1-p_1'}{1-\Gamma_1}\Gamma_2,\dots, \frac{1-p_1'}{1-\Gamma_1}\Gamma_d\right).
\end{equation}
Again, it is clear that $\v{p}'\succ_{\v{\Gamma}}\check{\v{p}}'$. Since 
\begin{equation}
    \check{p}_1'=p_1'\geq \gamma_1 \geq \Gamma_1,    
\end{equation}
and the remaining entries $\check{p}_i'$ are proportional to entries of the thermal distribution $\v{\Gamma}$, the thermomajorisation curve of $\check{\v{p}}'$ with respect to $\v{\Gamma}$ is defined by three points: $(0,0)$, $(\Gamma_1,p_1')$ and $(1,1)$. Thus, for the $x$ component $\Gamma_d$, its $y$ component is given by
\begin{equation}
    p_d''=\frac{\Gamma_d}{\Gamma_1}p_1'=\frac{(1-\gamma_1)\Gamma_d}{(1-\gamma_d)\Gamma_1}p_d +\frac{(\gamma_1-\gamma_d)\Gamma_d}{(1-\gamma_d)\Gamma_1}.
\end{equation}
This means that $\check{\v{p}}'$ can be transformed to $\v{p}''$ with the above occupation of the highest state $p_d''$ and the remaining occupations $p_i''$ simply chosen to be proportional to the entries of the thermal distribution $\v{\Gamma}$. Since thermomajorisation is transitive, we have \mbox{$\v{p}'\succ_{\v{\Gamma}} \v{p}''$}. Combining everything together and using again the transitivity property of thermomajorisation, we end up with Eq.~\eqref{eq:thermomajo_chain} (its half concerning $\v{p}$, $\v{p}'$ and $\v{p}''$) with $\v{p}''$ satisfying Eq.~\eqref{eq:pd_double_prime}. Finally, to show that $p_d''\geq p_d$ and $p_d''\leq \tilde{\Gamma}_d$, one simply needs to use the assumption that $\Gamma_d\leq p_d\leq \tilde{\Gamma}_d$.

Now, let us switch to the second part of the lemma. Consider a state $\check{\v{q}}$ obtained from $\v{q}$ by thermalising the levels \mbox{$\{2,\dots,d\}$} of the system with respect to the hot bath, i.e.,
\begin{equation}
    \check{\v{q}}=\left(q_1,\frac{1-q_1}{1-\Gamma_1}\Gamma_2,\dots, \frac{1-q_1}{1-\Gamma_1}\Gamma_d\right).
\end{equation}
Clearly, $\v{q}\succ_{\v{\Gamma}}\check{\v{q}}$, as thermalising a subset of states is a thermal operation. Since 
\begin{equation}
    \check{q}_1=q_1\geq \gamma_1 \geq \Gamma_1,    
\end{equation}
and the remaining entries $\check{q}_i$ are proportional to the entries of the thermal distribution $\v{\Gamma}$, the thermomajorisation curve of $\check{\v{q}}$ with respect to $\v{\Gamma}$ is defined by three points: $(0,0)$, $(\Gamma_1,q_1)$ and $(1,1)$. Thus, for the $x$ component $\Gamma_d$, its $y$ component is given by
\begin{equation}
    q_d'=\frac{\Gamma_d}{\Gamma_1}q_1.
\end{equation}
This means that $\check{\v{q}}$ can be transformed to $\v{q}'$ with the above occupation of the highest excited state $q_d'$ and the remaining occupations $q_i'$ simply chosen to be proportional to the entries of the thermal distribution $\v{\Gamma}$. Finally, since thermomajorisation is transitive, we have \mbox{$\v{q}\succ_{\v{\Gamma}} \v{q}'$}.

The next step is to start with $\v{q}'$ and thermalise levels $\{1,\dots,d-1\}$ with respect to the cold bath, resulting in
\begin{equation}
    \check{\v{q}}'=\left(\frac{1-q_d'}{1-\gamma_d}\gamma_1,\dots,\frac{1-q_d'}{1-\gamma_d}\gamma_{d-1},q_d'\right).
\end{equation}
Again, it is clear that $\v{q}'\succ_{\v{\gamma}}\check{\v{q}}'$. Since 
\begin{equation}
    \check{q}_d'=q_d'\geq \Gamma_d \geq \gamma_d,    
\end{equation}
and the remaining entries $\check{q}_i'$ are proportional to entries of the thermal distribution $\v{\gamma}$, the thermomajorisation curve of $\check{\v{q}}'$ with respect to $\v{\gamma}$ is defined by three points: $(0,0)$, $(\gamma_d,q_d')$ and $(1,1)$. Thus, for the $x$ component $\gamma_1$, its $y$ component is given by
\begin{equation}
    q_1''=q_d'+\frac{\gamma_1-\gamma_d}{1-\gamma_d}(1-q_d')=\frac{(1-\gamma_1)\Gamma_d}{(1-\gamma_d)\Gamma_1}q_1+\frac{\gamma_1-\gamma_d}{1-\gamma_d}.
\end{equation}
This means that $\check{\v{q}}'$ can be transformed to $\v{q}''$ with the above occupation of the ground state $q_1''$ and the remaining occupations $q_i''$ simply chosen to be proportional to the entries of the thermal distribution $\v{\gamma}$. Since thermomajorisation is transitive, we have \mbox{$\v{q}'\succ_{\v{\gamma}} \v{q}''$}. Combining everything together and using again the transitivity property of thermomajorisation, we end up with Eq.~\eqref{eq:thermomajo_chain} (its half concerning $\v{q}$, $\v{q}'$ and $\v{q}''$) with $\v{q}''$ satisfying Eq.~\eqref{eq:q1_double_prime}. Finally, to show that $q_1''\geq q_1$ and $q_1''\leq \tilde{\gamma}_1$, one simply needs to use the assumption that $\gamma_1\leq q_1\leq \tilde{\gamma}_1$.

\end{proof}


\section{Proof of Proposition~\ref{prop:ground}}
\label{app:ground}

The proof is based on the following crucial lemma.
\begin{lem}
    \label{lem:groundstate}
    Assume $\v{\gamma}$ and $\v{\Gamma}$ satisfy the requirements of Proposition~\ref{prop:ground} from Eq.~\eqref{eq:coldenough}, and the initial state $\v{p}$ satisfies
    \begin{equation}
        \label{eq:goodbetaorderp}
        \forall k:\quad  \frac{p_k}{\gamma_k} \geq \frac{p_{k+1}}{\gamma_{k+1}}.
    \end{equation}
    Then,
    \begin{equation}
        \v{p}\succ_{\v{\Gamma}} \v{q} \succ_{\v{\gamma}} \v{r},
    \end{equation}
    with
    \begin{subequations}
    \begin{align}
        \label{eq:q}
       q_1  =p_d+\frac{\Gamma_1-\Gamma_d}{\Gamma_{d-1}}p_{d-1},\quad & \quad
       q_k  =\frac{\Gamma_k}{1-\Gamma_1}(1-q_1),\\
       \label{eq:r}
       r_1  =1-\frac{1-\gamma_1}{\gamma_{1}}q_{1},\quad & \quad
       r_k  =\frac{\gamma_k}{\gamma_1}q_1,
    \end{align}
    \end{subequations}
    for $k>1$. Moreover,
    \begin{equation}
        \label{eq:goodbetaorderr}
        \forall k:\quad     \frac{r_k}{\gamma_k} \geq \frac{r_{k+1}}{\gamma_{k+1}}.
    \end{equation}
\end{lem}

For the sake of clarity of the argument, we defer the proof of the above lemma to the end of this appendix. The lemma tells us that a sequential interaction of the initial state $\v{p}$ with a hot and a cold bath can transform it to a final state $\v{r}$. Since the final state $\v{r}$ satisfies Eq.~\eqref{eq:goodbetaorderr}, one can apply Lemma~\ref{lem:groundstate} iteratively, starting from $\v{p}[1]=\v{\gamma}$. Thus, given a vector $\v{p}[n]$ satisfying Eq.~\eqref{eq:goodbetaorderp}, we can transform it to $\v{p}[n+1]$ also satisfying Eq.~\eqref{eq:goodbetaorderp} and with
\begin{align}
    \label{eq:p1n}
    p_1[n+1] &=1-\frac{1-\gamma_1}{\gamma_1}\left(p_d[n]+\frac{\Gamma_1-\Gamma_d}{\Gamma_{d-1}}p_{d-1}[n]\right),\\
    p_k[n+1] &=\frac{\gamma_k}{\gamma_1}\left(p_d[n]+\frac{\Gamma_1-\Gamma_d}{\Gamma_{d-1}}p_{d-1}[n]\right),
\end{align}
for $k>1$. Now, we have 
\begin{equation}
    p_d[n+1]+p_{d-1}[n+1]=\frac{\gamma_d+\gamma_{d-1}}{\gamma_1}\left(p_d[n]+\frac{\Gamma_1-\Gamma_d}{\Gamma_{d-1}}p_{d-1}[n]\right)\leq \frac{\gamma_d+\gamma_{d-1}}{\gamma_1}\left(p_d[n]+p_{d-1}[n]\right),
\end{equation}
where we have used the assumption $\Gamma_1<\Gamma_{d}+\Gamma_{d-1}$ from Eq.~\eqref{eq:coldenough}. We thus have
\begin{equation}
    p_d[n+1]+p_{d-1}[n+1] \leq \left(\frac{\gamma_d+\gamma_{d-1}}{\gamma_1}\right)^n (p_d[1]+p_{d-1}[1]),
\end{equation}
and because the exponentiated prefactor on the right hand side is strictly smaller than 1 (again, due to the assumption from Eq.~\eqref{eq:coldenough}), we get
\begin{equation}
    \lim_{n\to\infty} (p_d[n]+p_{d-1}[n]) = 0.
\end{equation}
Obviously both $p_d[n]$ and $p_{d-1}[n]$ are positive, and so they both independently tend to zero as $n\to\infty$. Through Eq.~\eqref{eq:p1n}, this then means that
\begin{equation}
    \lim_{n\to\infty} p_1[n]=1,
\end{equation}
which proves Proposition~\ref{prop:ground}. Note that because the convergence is exponential, the ground state is only achieved in the limit $n\to\infty$. The last remaining thing to show is to prove Lemma~\ref{lem:groundstate}.

\begin{proof}[Proof of Lemma~\ref{lem:groundstate}]
    For $\v{p}$ satisfying Eq.~\eqref{eq:goodbetaorderp}, it is straightforward to show that it also satisfies
    \begin{equation}
        \forall k:\quad  \frac{p_k}{\Gamma_k} \geq \frac{p_{k+1}}{\Gamma_{k+1}}.
    \end{equation}
    Given the above, the thermomajorisation curve of $\v{p}$ with respect to the hot bath $\v{\Gamma}$ is given by the following elbow points:
    \begin{equation}
        \{(0,0),~(\Gamma_1,p_1),\dots,~(1-\Gamma_d-\Gamma_{d-1},1-p_d-p_{d-1}),~(1-\Gamma_d,1-p_d),~(1,1)\}.
    \end{equation}
    Since, by assumption from Eq.~\eqref{eq:coldenough}, we have
    \begin{equation}
        1-\Gamma_d-\Gamma_{d-1}\leq 1-\Gamma_1\leq 1-\Gamma_d,
    \end{equation}
    the thermomajorisation curve of $\v{p}$ at the $x$ position $(1-\Gamma_1)$ takes the value
    \begin{equation}
        1-q_1 = 1-p_d-p_{d-1}+\frac{\Gamma_d+\Gamma_{d-1}-\Gamma_1}{\Gamma_{d-1}}p_{d-1}=1-p_d-\frac{\Gamma_1-\Gamma_d}{\Gamma_{d-1}}p_{d-1}.
    \end{equation}
    Clearly then, a piecewise linear concave curve defined by the following elbow points,
    \begin{equation}
        \{(0,0),~(1-\Gamma_1,1-q_1),~(1,1)\},
    \end{equation}
    lies below the thermomajorisation curve of $\v{p}$. As can be verified by direct calculation, such a curve is a thermomajorisation curve of a state $\v{q}$
    from Eq.~\eqref{eq:q}, which satisfies
    \begin{equation}
        \label{eq:goodbetaorderq}
        \forall k:\quad     \frac{q_k}{\Gamma_k} \leq \frac{q_{k+1}}{\Gamma_{k+1}}.
    \end{equation}
    This proves the first part of the lemma.

    To prove the second part, we start by noting that Eq.~\eqref{eq:goodbetaorderq} implies 
    \begin{equation}
        \forall k:\quad     \frac{q_k}{\gamma_k} \leq \frac{q_{k+1}}{\gamma_{k+1}}.
    \end{equation}
    Given the above, the thermomajorisation curve of $\v{q}$ with respect to the cold bath $\v{\gamma}$ is given by the following elbow points:
    \begin{equation}
        \{(0,0),~(\gamma_d,q_d),~(\gamma_d+\gamma_{d-1},q_d+q_{d-1}),\dots,~(1-\gamma_1,1-q_1),~(1,1)\}.
    \end{equation}
     Since, by assumption from Eq.~\eqref{eq:coldenough}, we have
    \begin{equation}
        1-\gamma_1\leq \gamma_1\leq 1,
    \end{equation}
    the thermomajorisation curve of $\v{q}$ at the $x$ position $\gamma_1$ takes the value
    \begin{equation}
        r_1 = 1-q_1+\frac{2\gamma_1-1}{\gamma_1}q_1=1-\frac{1-\gamma_1}{\gamma_1}q_1.
    \end{equation}
    Clearly then, a piecewise linear concave curve defined by the following elbow points,
    \begin{equation}
        \{(0,0),~(\gamma_1,r_1),~(1,1)\},
    \end{equation}
    lies below the thermomajorisation curve of $\v{q}$. As can be verified by direct calculation, such a curve is a thermomajorisation curve of a state $\v{r}$
    from Eq.~\eqref{eq:r}, which satisfies
    \begin{equation}
        \forall k:\quad     \frac{r_k}{\gamma_k} \geq \frac{r_{k+1}}{\gamma_{k+1}}.
    \end{equation}
    This completes the proof.
        
\end{proof}


\section{Maximally mutually coherent states}
\label{app:mutually_coh}


\subsection*{Existence and generation in three strokes}

Let us start by recalling a known result about the existence of maximally mutually coherent states exist for any two given bases and in any dimension (see Ref.~\cite{Idel_wolf}, whose results imply the statement; cf. also Refs.~\cite{korzekwa_jennings_rudolph} and~\cite{puchala_rudnicki}).

\begin{thm}[{\cite[Theorem 1]{korzekwa_jennings_rudolph}}]\label{lem:torus}
For any two bases $\{|j\rangle\}_{j=1}^d$ and $\{U^\dagger| j\rangle\}_{j=1}^d$ of a $d$-dimensional Hilbert space there exist at least $2^{d-1}$ states $|\psi_*\rangle$ that are unbiased in both these bases, that is, 
\begin{align}\label{eq:unbiased}
\left|\left\langle j|\psi_*\right\rangle\right|=\left|\left\langle j|U|\psi_*\right\rangle\right|=\frac{1}{\sqrt{d}}\;\;\;\text{for all}\;\;\;j\in\{1,\ldots,d\},
\end{align}
which, in turn, implies that
\begin{align}
\begin{aligned}
\left|\psi_*\right\rangle
=\frac{1}{\sqrt{d}}\sum_{j=1}^de^{i\alpha_j}\left| j\right\rangle
=\frac{1}{\sqrt{d}}\sum_{j=1}^de^{i\beta_j}U^\dagger\left| j\right\rangle.
\end{aligned}
\end{align}
Moreover, any state $|\psi_*\rangle$ satisfying Eq.~\eqref{eq:unbiased} is a~\textit{zero-noise, zero-disturbance} state, also called a \textit{mutually coherent} (or \textit{maximally mutually coherent}) state (see Ref.~\cite{puchala_rudnicki}). 
\end{thm}

\begin{rem}
The assertion of Theorem~\ref{lem:torus} follows from the fact that a great torus $T_{d-1}$ embedded in a complex projective space $\mathbb{C}P^{d-1}$ is \emph{non-displaceable} with respect to transformations by a unitary matrix $U\in\mathcal{U}_d(\mathbb{C})$ (cf.~\cite{cho}), meaning that the image of torus resulting from the action of $U$ must intersect the original torus, in at least $2^{d-1}$ points (for a more detailed explanation see, e.g., Ref.~\cite{puchala_rudnicki}). 
\end{rem}

Next, consider the problem of generating a maximally mutually coherent state by performing only three operations (preparation of $\ket{l}$ by Alice, transformation $U^{\dagger}D_1U$ by Bob, and transformation $D_2$ by Alice). From what has been said above, this task can be reduced to guaranteeing the existence of $l\in\{1,\ldots,d\}$ and $D_1,D_2\in\mathcal{DU}_d(\mathbb{C})$ such that $U^{\dagger}D_1U$ has a flat $l$-th column, i.e., 
\begin{align}\label{eq:flat}
\left|\left\langle m\right|U^{\dagger}D_1U\left|l\right\rangle\right|=\frac{1}{\sqrt{d}}\;\;\;\text{for all}\;\;\;m\in\{1,\ldots,d\},
\end{align}
and 
\begin{align}\label{eq:phases}
D_2U^{\dagger}D_1U\left|l\right\rangle=\frac{1}{\sqrt{d}}\sum_{j=1}^de^{i\alpha_j}\left|j\right\rangle=\left|\psi_*\right\rangle.
\end{align}
Note that $D_2$ is only responsible for fitting the phases $\alpha_j$, \hbox{$j\in\{1,\ldots,d\}$}.


\subsection*{Qubit case}

Before we proceed to the proof of Proposition~\ref{rem:necessary}, 
which deals with a general problem of generating a~maximally mutually coherent state using just three strokes of a coherence engine, let us first show how the problem is solved in the case of a qubit (simple heuristics have been already presented in Sec.~\ref{sec:coherence_qubit}, however this analytical proof can be more easily generalised to higher dimensions). 

\begin{prop}
\label{prop:qubit_conditions}
Suppose that Alice and Bob are restricted to the sets $F_A$, $F_B$ of their free states and  $\mathcal{F}_A$, $\mathcal{F}_B$ of their free operations, respectively. Let $U\in\mathcal{U}_2(\mathbb{C})$, appearing in the definition of $\F_B$ in Eq.~\eqref{def:B}, be of the general form
\begin{align}\label{def:U_unitary}
U=e^{i\phi}\left[
\begin{array}{ll}
e^{i\varphi_0}\cos(\varphi)&-e^{-i\varphi_1}\sin(\varphi)\\
e^{i\varphi_1}\sin(\varphi)&e^{-i\varphi_0}\cos(\varphi)
\end{array}
\right].
\end{align}
Then, the parties can generate a maximally mutually coherent state with only three resource engine's strokes if and only if $\varphi\in[\pi/8,3\pi/8]$.
\end{prop}

\begin{proof}
Let $D:=\text{diag}(e^{i\theta_0}, e^{i\theta_1})$, with $\theta_0,\theta_1\in[0,2\pi)$, be an arbitrary diagonal unitary matrix. Due to Eq.~\eqref{def:U_unitary}, we obviously have
\begin{align}\label{eq:obvious}
\begin{aligned}
\langle 0|U^{\dagger}DU|0\rangle=e^{i\theta_0}\cos^2(\varphi)+e^{i\theta_1}\sin^2(\varphi)\quad\text{and}\quad 
\langle 1|U^{\dagger}DU|0\rangle
=\sin(\varphi)\cos(\varphi)e^{i\left(\varphi_0+\varphi_1\right)}(-e^{i\theta_0}+e^{i\theta_1}).
\end{aligned}
\end{align}
It follows from the previous subsection that one can generate a maximally mutually coherent state with just three resource engine's strokes if and only if either
\begin{align}\label{case1}
\exists_{\theta_2,\theta_3\in(0,2\pi]}\quad\langle 0|U^{\dagger}DU|0\rangle=e^{i\theta_2}\frac{1}{\sqrt{2}}\;\;\;\text{and}\;\;\;\langle 1|U^{\dagger}DU|0\rangle=e^{i\theta_3}\frac{1}{\sqrt{2}}
\end{align}
or 
\begin{align}\label{case2}
\exists_{\theta_2,\theta_3\in(0,2\pi]}\quad\langle 0|U^{\dagger}DU|1\rangle=e^{i\theta_2}\frac{1}{\sqrt{2}}\;\;\;\text{and}\;\;\;\langle 1|U^{\dagger}DU|1\rangle=e^{i\theta_3}\frac{1}{\sqrt{2}}.
\end{align}
If we now compare Eq.~\eqref{eq:obvious} with Eq.~\eqref{case1} (the reasoning for Eq.~\eqref{case2} is analogous), we get
\begin{align}\label{eq:triangle}
\begin{aligned}
\exists_{\theta_0,\theta_1,\theta_2,\theta_3\in[0,2\pi)}\quad
e^{i\theta_0}\cos^2(\varphi)+e^{i\theta_1}\sin^2(\varphi)-\frac{e^{i\theta_2}}{\sqrt{2}}=0,\quad
\sin(\varphi)\cos(\varphi)e^{i\left(\varphi_0+\varphi_1\right)}\big(e^{i\theta_1}-e^{i\theta_0}\big)-\frac{e^{i\theta_3}}{\sqrt{2}}=0.
\end{aligned}
\end{align}

To verify when the above holds, let us consider two triangles (in the complex plane) with vertices:
\begin{align}
\begin{aligned}
(1):\; &\big\{0,-e^{i\theta_0}\sin(\varphi)\cos(\varphi),
-e^{i\theta_0}\sin(\varphi)\cos(\varphi)
+e^{i\theta_1}\sin(\varphi)\cos(\varphi)\big\},\\
(2):\; &\left\{0,e^{i\theta_0}\cos^2(\varphi),
e^{i\theta_0}\cos^2(\varphi)+e^{i\theta_1}\sin^2(\varphi)\right\},
\end{aligned}
\end{align}
where $\varphi,\theta_0,\theta_1\in[0,2\pi)$. 
One can easily observe that triangle (1) is isosceles with legs of length $\sin(\varphi)\cos(\varphi)$. 
The base of this triangle can have length $1/\sqrt{2}$ only if 
$2\sin(\varphi)\cos(\varphi)\geq 1\sqrt{2}$ (the triangle inequality), 
which yields 
\begin{align}\label{eq:condition_for_phi}
\frac{\pi}{8}\leq\varphi\leq\frac{3\pi}{8}.
\end{align}
If we denote the measure of the angles at the base by $\alpha$, then the third angle obviously has measure $\pi-2\alpha$. Note that for a fixed value of $\varphi$, where $\varphi\in[\pi/8,3\pi/8]$, there is only one possible value of $\alpha$, namely 
\begin{align}\label{def:alpha}
\alpha=\arccos \left(\frac{1}{2\sqrt{2}\sin(\varphi)\cos(\varphi)}\right). 
\end{align} 

Analysis similar to the above one (based on applying triangle inequalities) indicates that the third side of triangle (2) has length $1/\sqrt{2}$ if and only if $\varphi$ and $\alpha$ (being half of the angle between the sides of triangle (2), of lengths $\cos^2(\varphi)$ and $\sin^2(\varphi)$) satisfy Eqs.~\eqref{eq:condition_for_phi} and \eqref{def:alpha}, respectively. As a consequence, we obtain the following: for any $\varphi,\varphi_0,\varphi_1\in[0,2\pi)$, Eq.~\eqref{eq:triangle} holds if and only if \hbox{$\varphi\in[\pi/8,3\pi/8]$}, which completes the proof.
\end{proof}


\subsection*{Proof of Proposition~\ref{rem:necessary}}

Using the intuition from the case of a qubit, let us now conduct the proof of Proposition~\ref{rem:necessary}.

\begin{proof}[Proof of Proposition~\ref{rem:necessary}]
The necessary condition in the case of $d=2$ (namely, $\pi/8\leq \varphi \leq 3\pi/8$) follows just from the triangle inequalities. For an arbitrary $d\in\mathbb{N}\backslash\{1,2\}$ the reasoning is analogous (with the  difference that in higher dimensions the necessary conditions are no longer the sufficient ones). As already mentioned at the beginning of this appendix, to solve the given task, it suffices to find a matrix $U=(u_{i,j})_{i,j=1}^d\in\mathcal{U}_d(\mathbb{C})$ for which Eq. \eqref{eq:flat} is satisfied with some $|l\rangle$, $l\in\{1,\ldots,d\}$, and some $D_1\in\mathcal{DU}_n(\mathbb{C})$, i.e., 
\begin{align}
\exists_{l\in\{1,\ldots,d\}}\;\exists_{\kappa_1,\ldots,\kappa_d\in(0,2\pi]}\;\forall_{m\in\{1,\ldots,d\}}\quad \left\langle m\left|U^{\dagger}D_1U\right|l\right\rangle=e^{i\kappa_m}\frac{1}{\sqrt{d}}.
\end{align}
Writing it differently, we obtain (for $D=\text{diag}(e^{i\xi_1},\ldots,e^{i\xi_d})$)
\begin{align}
\exists_{l\in\{1,\ldots,d\}}\;\exists_{\kappa_1,\ldots,\kappa_d\in(0,2\pi]}\;\forall_{m\in\{1,\ldots,d\}}\quad
\sum_{j=1}^de^{i\xi_j}\bar{u}_{jm}u_{jl}=e^{i\kappa_m}\frac{1}{\sqrt{d}},
\end{align} 
which is kind of an equivalent of Eq.~\eqref{eq:triangle}. 
The generalised polygon inequalities then imply the assertion of Proposition~\ref{rem:necessary}. 
\end{proof}


\section{Proof of Proposition~\ref{prop:equiv_H1_H2}}
\label{app:equivalence}

In Sec.~\ref{sec:coherence_operations}, for any $U\in\mathcal{U}_d(\mathbb{C})$, we have defined a matrix $P_U$, whose entries are zero whenever the corresponding entries of $U$ are zero, otherwise they are equal to 1  (cf. Eq. \eqref{def:P_U}). Here, it is convenient to additionally introduce a~matrix $R_U=(r_{ij})_{i,j=1}^d$ over the field $\mathbb{R}_+$, which corresponds to $U$ in the following sense:
\begin{align}\label{def:R_U}
\begin{aligned}
r_{ij}=\left\{\begin{array}{ll}
0&\text{for }u_{ij}=0,\\
\text{some positive number $r_{ij}$}&\text{for }u_{ij}\neq 0.
\end{array}\right.
\end{aligned}
\end{align}

The aim of this appendix is to prove that hypotheses \hyperref[cnd:H1]{(H1)} and \hyperref[cnd:H2]{(H2)} are equivalent. For this, we shall first establish a few auxiliary facts.

\begin{lem}\label{lem:diag}
Let $ A = (a_{ij})_{i,j=1}^d$ and $D=\text{diag}(d_1,\ldots,d_d)$ with $d_i\neq 0$ for any $i\in\{1,\ldots,d\}$. Then, for $AD=(b_{ij})_{i,j=1}^d$ and any $k,l\in\{1,\ldots,d\}$ 
\begin{align}
a_{kl}=0\quad\Leftrightarrow\quad b_{kl}=0.
\end{align}
\end{lem}
\begin{proof}
Note that for any $k,l\in\{1,\ldots,d\}$ we have 
\hbox{$b_{kl}=a_{kl}d_l$}, 
which immediately yields the assertion of Lemma~\ref{lem:diag}.
\end{proof}

\begin{lem}\label{lem:zeros_1}
Let \hbox{$U=(u_{ij})_{i,j=1}^d, V=(v_{ij})_{i,j=1}^d\in\mathcal{U}_d(\mathbb{C})$}, and let $R_U=(r_{ij})_{i,j=1}^d,R_V=(s_{ij})_{i,j=1}^d$ be the corresponding matrices, with entries in $\mathbb{R}_+$, satisfying Eq.~\eqref{def:R_U}. Then, there exists \hbox{$D=(d_{ij})_{i,j=1}^d\in\mathcal{DU}_d(\mathbb{C})$} such that for all \hbox{$k,l\in\{1,\ldots,d\}$}
\begin{align}
\left(R_UR_V\right)_{kl}\neq 0\quad\Leftrightarrow\quad \left(UDV\right)_{kl}\neq 0.
\end{align}
\end{lem}
\begin{proof}
Let $D\in\mathcal{DU}_d(\mathbb{C})$, and let $k,l\in\{1,\ldots,d\}$ be arbitrary. We shall first prove the simpler direction, i.e., 
\begin{align}\label{implication_1}
(UDV)_{kl}\neq 0\quad\Rightarrow\quad \left(R_UR_V\right)_{kl}\neq 0.
\end{align}
Observe that 
$(UDV)_{kl}\neq 0$
if and only if 
$\sum_{m=1}^du_{km}d_mv_{ml}\neq 0$. 
This, however, implies that there exists at least one \hbox{$m_0\in\{1,\ldots,d\}$} such that $u_{km_0}v_{m_0l}\neq 0$. Further, referring to Eq.~\eqref{def:R_U}, we get that $r_{km_0}s_{m_0l}>0$. Keeping in mind that $r_{ij},s_{ij}\geq 0$ for all $i,j\in\{1,\ldots,d\}$, we finally obtain
$\left(R_UR_V\right)_{kl}=\sum_{m=1}^dr_{km}s_{ml}>0$, 
which completes the proof of Eq.~\eqref{implication_1}.

We now proceed to the harder direction, i.e., to proving that for any $(k,l)\in\{1,\ldots,d\}$
\begin{align}
    \left(R_UR_V\right)_{kl}\neq 0\quad\Rightarrow\quad (UDV)_{kl}\neq 0.
\end{align}
Let us define 
\begin{align}
    I=\left\{(k,l):\;k,l\in\{1,\ldots,d\},\,\left(R_UR_V\right)_{kl}\neq 0\right\}. 
\end{align}
We shall construct a matrix $D\in\mathcal{DU}_d(\mathbb{C})$ such that
\begin{align}\label{condition}
(UDV)_{kl}\neq 0\;\;\;\text{for all}\;\;\;(k,l)\in I.
\end{align}
Note that $\left(R_UR_V\right)_{kl}\neq 0$ if and only if $\sum_{m=1}^d r_{km}s_{ml}\neq 0$. 
Since the matrices $R_U$ and $R_V$ satisfy Eq.~\eqref{def:R_U}, we further get that there exists at least one $m_0\in\{1,\ldots,d\}$ such that $u_{km_0}v_{m_0l}\neq 0$.

For any $(k,l)\in I$ let us define a vector
\begin{align}
&\mathbf{w}^{(kl)}=\left(w^{(kl)}_m\right)_{m=1}^{d}\quad
\text{with}\;\;\; w_m^{(kl)}=u_{km}v_{ml}\;\;\;\text{for all}\;\;\;m\in\{1,\ldots,d\}.
\end{align}
Further, let us also introduce 
\begin{align}
\mathcal{W}=\left\{\mathbf{w}^{(kl)}:\;(k,l)\in I\right\},
\end{align}
the subsets $G_N$ of $\mathcal{W}$ given by
\begin{gather}
G_N=\left\{\mathbf{w}\in\mathcal{W}:\;w_1=\ldots=w_{N-1}=0,\;w_{N}\neq 0\right\}\;\;\;\text{for}\;\;\;N\in\{1,\ldots,d\},
\end{gather}
and the corresponding sets of indexes
\begin{align}
    I_N=\left\{(k,l)\in I:\;\mathbf{w}^{(kl)}\in G_N\right\}\;\;\;\text{for}\;\;\;N\in\{1,\ldots,d\}.
\end{align}
Note that $\{G_N\}_{N=1}^d$ and $\{I_N\}_{N=1}^d$ constitute partitions of the sets  $\mathcal{W}$ and $I$, respectively. 

The construction of the matrix $D\in\mathcal{DU}_d(\mathbb{C})$ for which Eq.~\eqref{condition} holds now proceeds as follows. Define 
\begin{align}
\label{def:D(1)}
&D^{(1)}:=\text{diag}\left(d_1^{(1)},\ldots,d_{d}^{(1)}\right)\quad\text{with}\;\;\;
d_m^{(1)}=\left\{
\begin{array}{ll}
e^{i\varphi_1}&\text{if }m=1\\
1&\text{if }m\neq 1
\end{array}
\right.,\;\;\;\varphi_1\in(0,2\pi).
\end{align}
If $(k,l)\in I_1$, then $\mathbf{w}^{(kl)}\in G_1$, and
\begin{align}
\sum_{m=1}^{d}w^{(kl)}_md_m^{(1)}
&=e^{i\varphi_1}w_1^{(kl)}+\sum_{m=2}^dw_m^{(kl)}=\left(e^{i\varphi_1}-1\right)w_1^{(kl)}+\sum_{m=1}^dw_m^{(kl)}.
\end{align}
Hence, for $(k,l)\in I_1$, we have
\begin{align}\label{varphi1}
\left(UD^{(1)}V\right)_{kl}=\sum_{m=1}^du_{km}d_m^{(1)}v_{ml}=\sum_{m=1}^{d}w^{(kl)}_md_m^{(1)}\neq 0
\end{align}
if and only if
\begin{align}\label{varphi1_choice}
e^{i\varphi_1}\neq 1-\frac{1}{w_1^{(kl)}}\sum_{m=1}^dw_m^{(kl)}.
\end{align}
Note that $w_1^{(kl)}\neq 0$, since $\mathbf{w}^{(kl)}\in G_1$. 

For any $(k,l)\in I_1$ there is at most one value \hbox{$\varphi_1\in(0,2\pi)$} which violates Eq.~\eqref{varphi1_choice}. Since $I_1$ is a~finite set, we see that the set 
\begin{align}
F_1=\left\{\varphi_1\in(0,2\pi):\;\exists_{(k,l)\in I_1}\; e^{i\varphi_1}= 1-\sum_{m=1}^d\frac{w_m^{(kl)}}{w_1^{(kl)}}\right\}
\end{align}
is also finite. We can therefore easily choose \hbox{$\varphi_1\in(0,2\pi)$} such that Eq.~\eqref{varphi1} holds for all $(k,l)\in I_1$. In particular, it suffices to take any number $\varphi_1$ from the uncountable set
$F_1^{\prime}=(0,2\pi)\backslash F_1$.

Now, let $N\in\{1,\ldots,d-1\}$, and suppose that the matrix $D^{(N)}$ satisfies
\begin{align}\label{condition_N}
\left(UD^{(N)}V\right)_{kl}\neq 0\;\;\;\text{for all}\;\;\;(k,l)\in I_1\cup\ldots\cup I_N.
\end{align}
Moreover, assume that $D^{(N)}$ has the following form:
\begin{align}
\label{def:D(N)}
&D^{(N)}:=\text{diag}\left(d_1^{(N)},\ldots,d_d^{(N)}\right)\quad\text{with}\;\;\;
d_m^{(N)}=\left\{
\begin{array}{ll}
e^{i\varphi_m}&\text{if }m\in\{1,\ldots,N\}\\
1&\text{otherwise }
\end{array}
\right.,
\end{align}
where $\varphi_1,\ldots,\varphi_N\in (0,2\pi)$ are some constants. Then we can introduce matrices
\begin{align}\label{def:D(N+1)}
\begin{aligned}
D^{(N+1)}:=\text{diag}\left(d_1^{(N+1)},\ldots,d_d^{(N+1)}\right)\;\;\;
\text{with}\;\;\;
d_m^{(N+1)}=\left\{
\begin{array}{ll}
e^{i\varphi_{N+1}}&\text{if }m=N+1\\
d_m^{(N)}&\text{if }m\neq N+1
\end{array}
\right.,\;\;\;
\text{where}\;\;\;\varphi_{N+1}\in(0,2\pi),
\end{aligned}
\end{align}
The aim now is to prove that the constant $\varphi_{N+1}$ can be always chosen in such a way that 
\begin{align}\label{aim}
\left(UD^{(N+1)}V\right)_{kl}\neq 0\;\;\;\text{for all}\;\;\;(k,l)\in \bigcup_{u=1}^{N+1}I_u.
\end{align}

Define the sets $F_{(N+1),u}$ with $u\in\{1,\ldots,N+1\}$ as
\begin{align}
\begin{aligned}
F_{(N+1),u}
=\left\{\varphi_{N+1}\in(0,2\pi):\;\exists_{D^{(N+1)}\text{given by Eq.~\eqref{def:D(N+1)}}}\exists_{(k,l)\in I_u}\;\left(UD^{(N+1)}V\right)_{kl}= 0\right\}.
\end{aligned}
\end{align}
Let us prove that for any $u\in\{1,\ldots,N+1\}$ the set $F_{(N+1),u}$ is finite. Fix $u\in\{1,\ldots,N+1\}$ and $(k,l)\in I_u$ arbitrarily. We have
\begin{align}
\left(UD^{(N+1)}V\right)_{kl}
&=\sum_{m=1}^dw_m^{(kl)}d_m^{(N+1)}
=\sum_{m=1}^{N+1}\left(e^{i\varphi_m}-1\right)w_m^{(kl)}+\sum_{m=1}^dw_m^{(kl)},
\end{align}
which is equal to zero if and only if 
\begin{align}
\label{condition_for_F_N+1}
\left(e^{i\varphi_{N+1}}-1\right)w_{N+1}^{(kl)}&=
-\sum_{m=1}^dw_m^{(kl)}
-\sum_{m=1}^{N}\left(e^{i\varphi_m}-1\right)w_m^{(kl)}.
\end{align}
If $w_{N+1}^{(kl)}\neq 0$, then Eq.~\eqref{condition_for_F_N+1} is equivalent to 
\begin{align}
e^{i\varphi_{N+1}}=
1-\sum_{m=1}^{d}\frac{w_m^{(kl)}}{w_{N+1}^{(kl)}}
-\sum_{m=1}^{N}\left(e^{i\varphi_m}-1\right)\frac{w_m^{(kl)}}{w_{N+1}^{(kl)}},
\end{align}
which holds for at most one $\varphi_{N+1}\in(0,2\pi)$. On the other hand, if $w_{N+1}=0$, then it follows from Eq.~\eqref{condition_for_F_N+1} that
\begin{align}
0&=-\sum_{m=1}^dw_m^{(kl)}
-\sum_{m=1}^{N}\left(e^{i\varphi_m}-1\right)w_m^{(kl)}
=-\left(UD^{(N)}V\right)_{kl}.
\end{align}
This, however, contradicts our assumption from Eq.~\eqref{condition_N}. Hence, we see that for any $u\in\{1,\ldots,N+1\}$
\begin{align}\label{def:F_N+1u}
\begin{aligned}
F_{(N+1),u}=\Bigg\{&\varphi_{N+1}\in(0,2\pi):\;\;\exists_{(k,l)\in I_u}\;w_{N+1}^{(kl)}\neq 0\;\;\text{and}\;\;
 e^{i\varphi_{N+1}}=1-\sum_{m=1}^{d}\frac{w_m^{(kl)}}{w_{N+1}^{(kl)}}
-\sum_{m=1}^{N}\left(e^{i\varphi_m}-1\right)\frac{w_m^{(kl)}}{w_{N+1}^{(kl)}}\Bigg\},
\end{aligned}
\end{align}
and, since $I_u$ is finite for any $u\in\{1,\ldots,N+1\}$, then so is $F_{(N+1),u}$. 

Now, choosing 
\begin{align}
\varphi_{N+1}\in F_{N+1}^{\prime}=(0,2\pi)\backslash F_{N+1},
\;\;\;\text{where}\;\;\;
F_{N+1}=\bigcup_{u=1}^{N+1}F_{(N+1),u},
\end{align}
we can guarantee that condition from Eq.~\eqref{aim} is fulfilled. 

Finally, recalling that $I=\bigcup_{m=1}^nI_m$, we can define 
\begin{align}
\label{def:D final}
D&:=\text{diag}\left(d_1,\ldots,d_d\right)
\quad\text{with}\;\;\;
d_m=e^{i\varphi_m},\;\;\;\varphi_m\in F_m^{\prime}\;\;\;\text{for}\;\;\;m\in\{1,\ldots,d\},
\end{align}
where 
\begin{align}
F_m^{\prime}=(0,2\pi)\backslash F_m,\;\;\;F_{m}=\bigcup_{u=1}^mF_{m,u}
\end{align}
and any $F_{m,u}$ is defined as in Eq.~\eqref{def:F_N+1u}. The above choice of $\varphi_1,\ldots,\varphi_d$ is always possible (if done inductively on $d$), since the sets $F_1^{\prime},\ldots,F_d^{\prime}$ are uncountable.

\end{proof}

As a corollary of Lemmas~\ref{lem:diag}~and~\ref{lem:zeros_1} we obtain the following statement.

\begin{cor}\label{concluding_corollary}
Let $U\in\mathcal{U}_d(\mathbb{C})$, and let $R_U$ be an arbitrary matrix over the field $\mathbb{R}_+$ satisfying Eq.~\eqref{def:R_U} (in particular, it can be $P_U$, given by Eq.~\eqref{def:P_U}). Then there exist matrices \hbox{$D_1, D_2\in\mathcal{DU}_d(\mathbb{C})$} such that for any $k,l\in\{1,\ldots,d\}$
\begin{align}
\left(R_U^TR_U\right)_{kl}=0
\quad\Leftrightarrow\quad
\left(D_1U^{\dagger}D_2U\right)_{kl}=0.
\end{align}
\end{cor}

To finally prove the equivalence of hypotheses \hyperref[cnd:H1]{(H1)} and \hyperref[cnd:H2]{(H2)} it suffices to establish the following lemma.

\begin{lem}\label{lem:last_step}
Let $U\in\mathcal{U}_d(\mathbb{C})$. For any $M\in\mathbb{N}$ there exist matrices $D_1,\ldots,D_{2M}\in\mathcal{DU}_d(\mathbb{C})$ such that for any \hbox{$k,l\in\{1,\ldots,d\}$}
\begin{align}
    \label{eq:lem:21}
    \left(\left(P_U^TP_U\right)^M\right)_{kl}=0\quad\Leftrightarrow \quad \left(D_1U{\dagger}D_2U\ldots D_{2M-1}U^{\dagger}D_{2M}U\right)_{kl}=0.
\end{align}
\end{lem}
\begin{proof}

First of all, note that the assertion of Lemma~\ref{lem:last_step}  for $M=1$ immediately follows from Corollary~\ref{concluding_corollary}. To establish it for all $M\in\mathbb{N}$, we shall conduct an inductive proof. Suppose that the assertion of Lemma~\ref{lem:last_step} holds for some $M\in\mathbb{N}$, i.e., there exist matrices \hbox{$D_1,\ldots,D_{2M}\in\mathcal{DU}_d(\mathbb{C})$} such that for any \hbox{$k,l\in\{1,\ldots,d\}$} Eq.~\eqref{eq:lem:21} holds. We will prove that, upon this assumption, there exist matrices $D_{2M+1},D_{2M+2}\in\mathcal{DU}_d(\mathbb{C})$ such that for all \hbox{$k,l\in\{1,\ldots,d\}$} we have
\begin{align}
\label{induction}
    &\left(\left(P_U^TP_U\right)^{M+1}\right)_{kl}=0\quad\Leftrightarrow\quad
    \left(D_1U{\dagger}D_2U\ldots D_{2M+1}U^{\dagger}D_{2M+2}U\right)_{kl}=0.
\end{align}

According to Corollary~\ref{concluding_corollary} (applied for $P_U$), there exists $D\in\mathcal{DU}_d(\mathbb{C})$ such that, for all $k,l\in\{1,\ldots,d\}$, 
\mbox{$\left(P_U^TP_U\right)_{kl}=0$}
if and only if $\left(\mathbbm{1}U^{\dagger}DU\right)_{kl}=0$. 
Now, keeping in mind the inductive assumption and the above, set
\begin{align}
V:=D_1U^{\dagger}D_2U\ldots D_{2M-1}U^{\dagger}D_{2M}U\;\;\;\text{and}\;\;\;W:=U^{\dagger}DU,
\end{align}
and note that the matrices $R_V=\left(P^T_UP_U\right)^{M}$ and $R_W=P^T_UP_U$
satisfy Eq.~\eqref{def:R_U} for $V$ and $W$, respectively. Hence, using Lemma~\ref{lem:zeros_1}, we obtain that there exists $\tilde{D}\in\mathcal{DU}_d(\mathbb{C})$ for which the equivalence
\begin{align}
&\left(\left(P^T_UP_U\right)^MP_U^TP_U\right)_{kl}
=\left(R_VR_W\right)_{kl}=0\quad\Leftrightarrow\quad
\left(V\tilde{D}W\right)_{kl}=0
\end{align}
holds for all $k,l\in\{1,\ldots,d\}$. This, however, implies Eq.~\eqref{induction} for $D_{2M+1}=\tilde{D}$ and $D_{2M+2}=D$, since
\begin{align}
V\tilde{D}W=\left(D_1U^{\dagger}\ldots D_{2M}U\right)\tilde{D}\left(U^{\dagger}DU\right).
\end{align}
\end{proof}


\section{Relations to Markov chains}
\label{app:markov}

For the convenience of the reader, we summarise here certain definitions and facts concerning finite Markov chains, to which we refer in Sec.~\ref{sec:coherence_operations}. 
In particular, we focus on irreducible and aperiodic Markov chains (for more details see, e.g., Refs.~\cite{haggstrom_2002,markov}). 

Consider a time-homogenous Markov chain \hbox{$\mathbf{X}=(X_n)_{n\in\mathbb{N}_0}$} with a countable state space $\Sigma$ (equipped with a~countably generated $\sigma$-field $\mathcal{B}(\Sigma)$ on $\Sigma$), and let 
\begin{align}
\begin{aligned}
T=\{T_{ij}:\,i,j\in \Sigma\}\;\;\;
\text{with}\;\;\;T_{ij}=\mathbb{P}(X_{n+1}=j|X_n=i)\;\;\;\text{for}\;\;\;n\in\mathbb{N}_0
\end{aligned}
\end{align}
be its transition matrix. The $m$-th step transition matrices \hbox{$T(m)=\{T_{ij}(m):\,i,j\in \Sigma\}$} are given by
\begin{align}\label{def:mth_step}
\begin{aligned}
T(0)=\mathbbm{1}\;\;\;
\text{and}\;\;\;T_{ij}(m)=\sum_{k\in \Sigma}T_{ik}T_{kj}(m-1)\;\;\;\text{for}\;\;\;m\in\mathbb{N},
\end{aligned}
\end{align}
and they describe the probabilities of getting from one state to another (or the same) in exactly $m$ steps. 

We say that a state $i\in \Sigma$ is \emph{accessible} from a state \hbox{$j\in\Sigma$} if and only if there exists some non-negative constant $m$ such that $T_{ij}(m)>0$. If $i$ and $j$ are accessible from each other, we call them \emph{communicating}. It shall be noted that the communication relation is an equivalence relation. 
A Markov chain $\mathbf{X}$ is called \emph{irreducible} if and only if its state space constitutes one communication class. 
In the graph-representation of the chain, for communicating states $i$ and $j$ there are directed paths from $i$ to $j$ and from $j$ to $i$. As a~consequence, a finite Markov chain is irreducible if and only if its graph
representation is a strongly connected graph (a directed graph is \emph{strongly connected} if every vertex is reachable from every other vertex). 
The \emph{period} $d(i)$ of a state $i\in \Sigma$ of a~Markov chain $\mathbf{X}$ is given by
$d(i)=\text{GCD}\{n\geq 1:\,T_{ii}(n)>0\}$. 
We call a state $i$ \emph{aperiodic} if $d(i)=1$. 
A Markov chain is \emph{aperiodic} if and only if all its states are aperiodic.

The following well-known facts (whose proofs can be found, e.g., in Ref.~\cite{markov}; cf. Corollaries 3.3.3 and 3.3.4 therein) shall be useful for our purposes.

\begin{prop}\label{prop:_irreducible_aperiodic}
Let $\mathbf{X}$ be an irreducible and aperiodic Markov chain with finite
state space and transition matrix~$T$. Then, there exists a finite constant $M$ such that for all $m\geq M$
\begin{align}
T_{ij}(m)>0\;\;\;\text{for all states}\;\;\;i,j\in \Sigma,
\end{align}
meaning that $T(m)$ is a matrix with all non-zero entries, and so any two states are communicating.
\end{prop}

Note that the assertion of Proposition~\ref{prop:_irreducible_aperiodic} might not hold if we abandon the assumption that a~Markov chain is aperiodic. Indeed, there exist irreducible Markov chains such that all matrices $T(m)$ have zero entries (and kind of a block structure). An exemplary transition matrix of such a chain is
\begin{align}
T=\left[
\begin{array}{cccc}
0&0&1/2&1/2\\
1&0&0&0\\
0&1&0&0\\
0&1&0&0
\end{array}
\right].
\end{align}

\begin{prop}\label{prop:loop}
If a Markov chain is irreducible, and additionally its directed graph has a~vertex with a loop (which is equivalent to requiring that the transition matrix $T$ of this Markov chain has at least one nonzero diagonal element), then it is also aperiodic.
\end{prop}

As a corollary of Propositions~\ref{prop:_irreducible_aperiodic}~and~\ref{prop:loop} we obtain the following:
\begin{prop}\label{prop:graph}
If a Markov chain has a directed graph which is strongly connected and has at least one vertex with a loop, then it is both irreducible and aperiodic.
\end{prop}

Using Proposition~\ref{prop:graph}, we can construct a simple example illustrating how to directly approach the problem of estimating the value $m\in\mathbb{N}$ describing the number of steps needed to generate a matrix $T(m)$ with all non-zero entries.

\begin{ex}\label{ex:graph}
Let $U\in\mathcal{U}_d(\mathbb{C})$, and let $W$ be its unistochastic matrix, that is, $W=U\circ \bar{U}$, where $\circ$ represents the entry-wise product (also known as Hadamard or Schur product). 
Consider a finite Markov chain described by a~transition matrix $T=W^TW$ and a directed graph $\mathcal{D}$ with vertices $\{v_1,\ldots,v_d\}$. 
If we assume that $\mathcal{D}$ is strongly connected and one of its vertices, say $v_{i_0}$, has a loop, then, according to Proposition~\ref{prop:graph}, we obtain the following: 
\begin{align}
\exists_{M\in\mathbb{N}}\;\forall_{i,j\in\mathbb\{1,\ldots,d\}}\;\forall_{m\geq M}\;\;\;T_{ij}(m)>0,
\end{align} 
where $T(m)$ denotes here a $m$-th step transition matrix, defined in Eq.~\eqref{def:mth_step}. Willing to find a constant $M$ for which the above condition holds, we may proceed as follows:
\begin{enumerate}
\item Since the graph $\mathcal{D}$ is strongly connected, for any $i,j\in\{1,\ldots,d\}$, there exist
\begin{align}
\alpha_i:=\min\left\{\alpha\in\mathbb{N}:\;T_{ii_0}(\alpha)>0\right\}
\;\;\;\text{and}\;\;\;
\beta_j:=\min\left\{\beta\in\mathbb{N}:\;T_{i_0j}(\beta)>0\right\},
\end{align}
which are finite. 
\item Note that for all $i,j\in\{1,\ldots,d\}$ and  any \hbox{$m\geq \alpha_i+\beta_j$}, we have
$T_{ij}(m)>0$. 
Indeed, letting $m=\alpha_i+\beta_j+k$ with some $k\in\mathbb{N}$, and recalling that $v_{i_0}$ has a loop, we get
$T_{ij}(m)\geq T_{ii_0}(\alpha_i)\left(T_{i_0i_0}(1)\right)^kT_{i_0j}(\beta_j)>0$.
\item Define
$M:=2\max_{i\in\{1,\ldots,d\}}\alpha_i$. 
Since the matrix $T=W^TW$ is symmetric, we have
\begin{align}
    M=\max_{i\in\{1,\ldots,d\}}\alpha_i+\max_{j\in\{1,\ldots,d\}}\beta_j.
\end{align}
Hence, according to Step 2, we know that for all $i,j\in\{1,\ldots,d\}$ and  any $m\geq M$
\begin{align}
    T_{ij}(m)
\geq T_{ii_0}(\alpha_i)
\left(T_{i_0i_0}(1)\right)^{m-\alpha_i-\beta_j}T_{i_0j}(\beta)>0.
\end{align}
\end{enumerate}
\end{ex}


\section{Proof of Proposition~\ref{prop:lower_bound}}
\label{app:proofs-lower-bound}

In order to prove Proposition~\ref{prop:lower_bound} (identifying the lower bound on the number of coherence engine's strokes
needed to produce all operations), let us first formulate and prove two auxiliary lemmas.

\begin{lem}\label{lem:X^TX}
Let $U\in\mathcal{U}_d(\mathbb{C})$, and let $X_U=(x_{ij})_{i,j=1}^d$ be such that
\begin{align}\label{def:X_U}
x_{ij}=\left|u_{ij}\right|\;\;\;\text{for any}\;\;\;i,j\in\{1,\ldots,d\}.
\end{align}
Then, for any $M\in\mathbb{N}$, we have
\begin{align}
\begin{aligned}
\max_{D_1,\ldots,D_{2M}\in\mathcal{DU}_d(\mathbb{C})}\left|\left\langle a\left|D_{2M}U^{\dagger}D_{2M-1}U\ldots D_2U^{\dagger}D_1U\right|b\right\rangle\right|
\leq\left\langle a\left|\left(X_U^TX_U\right)^M\right|b\right\rangle\;\;\;\text{for all}\;\;\;a,b\in\{1,\ldots,d\}.
\end{aligned}
\end{align}
\end{lem}
\begin{proof}
First of all, note that for any \hbox{$a,b\in\{1,\ldots,d\}$} we have
\begin{align}
\begin{aligned}
\max_{D_1,D_2\in\mathcal{DU}_d(\mathbb{C})}
\left|\left\langle a
\left|D_2U^{\dagger}D_1U\right|
b\right\rangle\right|
=\max_{\theta_1,\ldots,\theta_d,\gamma_1,\ldots,\gamma_d\in\mathbb{R}}
\left|\sum_{j=1}^d \bar{u}_{ja}e^{i\theta_j}u_{jb}e^{i\gamma_j}\right|
\leq\sum_{j=1}^d
\left|\bar{u}_{ja}\right|\left| u_{jb}\right|
=\left\langle a\left|\left(X_U^TX_U\right)\right|b\right\rangle,
\end{aligned}
\end{align}
where $X_U=(x_{ij})_{i,j=1}^d$ is given by Eq.~\eqref{def:X_U}.

The proof then proceeds by induction on $M\in\mathbb{N}$. If for any $M\in\mathbb{N}$ and all $a,b\in\{1,\ldots,d\}$ 
\begin{align}
\begin{aligned}
\max_{D_1\ldots,D_{2M}\in\mathcal{DU}_d(\mathbb{C})}\left|\left\langle a\left|D_{2M}U^{\dagger}D_{2M-1}U\ldots D_2U^{\dagger}D_1U\right|b\right\rangle\right|
\leq\left\langle a\left|\left(X_U^TX_U\right)^M\right|b\right\rangle,
\end{aligned}
\end{align}
then also
\begin{align}
\begin{aligned}
&\max_{D_1,\ldots,D_{2M+2}\in\mathcal{DU}_d(\mathbb{C})}\left|\left\langle a\left|D_{2M+2}U^{\dagger}D_{2M+1}U\ldots D_1U\right|b\right\rangle\right|\\
&\qquad\qquad=
\max_{D_1,\ldots,D_{2M+2}\in\mathcal{DU}_d(\mathbb{C})}\left|\sum_{k=1}^d\left\langle a\left|D_{2M+2}U^{\dagger}D_{2M+1}U\right|k\right\rangle\left\langle k\left|D_{2M}U^{\dagger}\ldots D_1U\right|b\right\rangle\right|\\
&\qquad\qquad\leq \sum_{k=1}^d
\left(\max_{D_{1},D_{2}\in\mathcal{DU}_d(\mathbb{C})}\left|\left\langle a\left|D_{2}U^{\dagger}D_{1}U\right|k\right\rangle\right|\right)
\left(\max_{D_1,\ldots,D_{2M}\in\mathcal{DU}_d(\mathbb{C})}\left|\left\langle k\left|D_{2M}U^{\dagger}\ldots D_1U\right|b\right\rangle\right|\right)\\
&\qquad\qquad\leq \sum_{k=1}^d\left\langle a\left|\left(X_U^TX_U\right)\right|k\right\rangle\left\langle k\left|\left(X_U^TX_U\right)^M\right|b\right\rangle
=\left\langle a\left|\left(X_U^TX_U\right)^{M+1}\right|b\right\rangle.
\end{aligned}
\end{align}
\end{proof}

Imagine that $U\in\mathcal{U}_d(\mathbb{C})$, appearing in the definition of $\F_B$, is \textit{close to} a diagonal matrix (meaning that the modulus of non-diagonal elements of $U$ are \textit{small}). Our intuition then tells us that Alice and Bob have to communicate many times before they manage to generate an arbitrary unitary matrix, and this is reflected in the following lemma.

\begin{lem}\label{lem:VDW}
Let $d\in\mathbb{N}\backslash\{1\}$, and let \hbox{$V=(v_{ij})_{i,j=1}^d,W=(w_{ij})_{i,j=1}^d\in\mathcal{U}_d(\mathbb{C})$} be such that for all $i,j\in\{1,\ldots,d\}$, $\;i\neq j$, we have
\begin{align}
\begin{aligned}
\left|v_{ij}\right|\leq v 
\;\;\;\text{for some}\;\;\;v\in\mathbb{R}_+
\;\;\;\text{and}\;\;\;
\left|w_{ij}\right|\leq w \;\;\;\text{for some}\;\;\;w\in\mathbb{R}_+.
\end{aligned}
\end{align}
Then
\begin{align}
\begin{aligned}
\max_{a,b\in\{1,\ldots,d\}:\;a\neq b}\max_{D\in\mathcal{DU}_d(\mathbb{C})}\left|\left\langle a\left|VDW\right|b\right\rangle\right|
\leq (d-2)vw+v+w.
\end{aligned}
\end{align}
\end{lem}

\begin{proof}
We have
\begin{align}
\begin{aligned}
\max_{a,b\in\{1,\ldots,d\}:\;a\neq b}
\max_{D\in\mathcal{DU}_d(\mathbb{C})}
\left|\left\langle a\left|VDW\right|b\right\rangle\right|
&=\max_{a,b\in\{1,\ldots,d\}:\;a\neq b}\max_{\theta_1,\ldots,\theta_d\in\mathbb{R}}
\left|\sum_{j=1}^dv_{aj}e^{i\theta_j}w_{jb}\right|\\
&\leq\max_{ab\in\{1,\ldots,d\}:\;a\neq b}\sum_{j=1}^d\left|v_{aj}\right|\left|w_{jb}\right|\\
&=\max_{a,b\in\{1,\ldots,d\}:\;a\neq b}
\Bigg(\sum_{j\in\{1,\ldots,d\}\backslash\{a,b\}}\left|v_{aj}\right|\left|w_{jb}\right|
+\left|v_{aa}\right|\left|w_{ab}\right|
+\left|v_{ab}\right|\left|w_{bb}\right|\Bigg)\\
&\leq (d-2)vw+v+w.
\end{aligned}
\end{align}
\end{proof}

We are now ready to give the proof of Proposition~\ref{prop:lower_bound}.

\begin{proof}[Proof of Proposition~\ref{prop:lower_bound}]
For any $M\in\mathbb{N}$ and any  $U\in\mathcal{U}_d(\mathbb{C})$,
let us introduce matrices
\begin{align}
V_{U}^{(D_1,\ldots,D_{2M})}=D_{2M}U^{\dagger}D_{2M-1}U\ldots D_2U^{\dagger}D_1U
\end{align}
\begin{align}
v_{U}^{(D_1,\ldots,D_{2M})}:=\max_{i,j\in\{1,\ldots,d\},\;i\neq j}
\left|\left\langle i\left|V_{U}^{(D_1,\ldots,D_{2M})}\right| j\right\rangle\right|.
\end{align}
For all $M\in\mathbb{N}\backslash\{1\}$, we then have
\begin{align}
\begin{aligned}
&\max_{a,b\in\{1,\ldots,d\}:\;a\neq b}\max_{D_1,\ldots,D_{2M}\in\mathcal{DU}_d(\mathbb{C})}\left|\left\langle a\left|D_{2M}U^{\dagger}D_{2M-1}U\ldots D_2U^{\dagger}D_1U\right|b\right\rangle\right|\\
&\qquad\qquad\qquad=\max_{a,b\in\{1,\ldots,d\}:\;a\neq b}\max_{D_1,\ldots,D_{2M}\in\mathcal{DU}_d(\mathbb{C})}
\max_{D\in\mathcal{DU}_d(\mathbb{C})}
\left|\left\langle a\left|V_{U}^{(D_3,\ldots,D_{2M})}DV_{U}^{(D_1,D_2)}\right|b\right\rangle\right|,
\end{aligned}
\end{align}
which, due to Lemma~\ref{lem:VDW} applied to $V=V_{U}^{(D_3,\ldots,D_{2M})}$ and $W=V_{U}^{(D_1,D_2)}$, entails the following:
\begin{align}
\begin{aligned}
&\max_{a,b\in\{1,\ldots,d\}:\;a\neq b}\max_{D_1,\ldots,D_{2M}\in\mathcal{DU}_d(\mathbb{C})}
\left|\left\langle a\left|D_{2M}U^{\dagger}D_{2M-1}U\ldots D_2U^{\dagger}D_1U\right|b\right\rangle\right|\\
&\qquad\qquad\qquad\leq\max_{D_1,\ldots,D_{2M}\in\mathcal{DU}_d(\mathbb{C})} \left((d-2)v_U^{(D_3,\ldots,D_{2M})}v_U^{(D_1,D_2)}
+v_U^{(D_3,\ldots,D_{2M})}+v_U^{(D_1,D_2)}\right).
\end{aligned}
\end{align}
Further, for any $M\in\mathbb{N}$ and any $U\in\mathcal{U}_d(\mathbb{C})$ define
\begin{align}
v_{U,M}:=\max_{D_1,\ldots,D_{2M}\in\mathcal{DU}_d(\mathbb{C})}v_U^{(D_1,\ldots,D_{2M})},
\end{align}
and observe that $v_{U,1}=c_U$ with $c_U$ given by Eq.~\eqref{def:c_U}. Moreover, for any $M\in\mathbb{N}\backslash\{1,2\}$ and any $U\in\mathcal{U}_d(\mathbb{C})$, we get \begin{align}
v_{U,M}\leq v_{U,M-1}\left((d-2)c_U+1\right)+c_U,
\end{align}
and so
\begin{align}\label{def:v_UN}
v_{U,M}\leq
\frac{\left((d-2)c_U+1\right)^M-1}{d-2}\;\;\;\text{for any}\;\;\;d\in\mathbb{N}\backslash\{1,2\}.
\end{align}

Now, observe that whenever the right hand side of Eq. \eqref{def:v_UN} is smaller than 1, then obviously also 
\begin{align}
\begin{aligned}
\max_{a,b\in\{1,\ldots,d\}:\;a\neq b}
\max_{D_1,\ldots,D_{2M}\in\mathcal{DU}_d(\mathbb{C})}
\left|\left\langle a\left|D_{2M}U^{\dagger}D_{2M-1}U\ldots D_2U^{\dagger}D_1U\right|b\right\rangle\right|<1,
\end{aligned}
\end{align}
which, in turn, implies that Alice and Bob are not able to generate any permutation matrix within $N=2M$ steps (i.e., $M$ Alice's strokes and $M$ Bob's strokes). As a consequence, we may claim that, in order to generate an arbitrary unitary matrix of order $d$, both Alice and Bob need to perform at least $M_{\text{min}}$ operations, with $M_{\text{min}}$ given as the smallest natural number satisfying
\begin{align}\label{def:N_min}
\frac{\left((d-2)c_U+1\right)^{M_{\text{min}}}-1}{d-2}\geq 1\;\;\;\text{for}\;\;\;d\in\mathbb{N}\backslash\{1,2\}.
\end{align}
Finally, we get 
\begin{align}
M_{\text{min}}\geq \frac{\log\left(d-1\right)}{\log\left((d-2)c_U+1\right)}\;\;\;\text{for}\;\;\;d\in\mathbb{N}\backslash\{1,2\}.
\end{align}
Since the minimal number of strokes is given by $N=2M_{\text{min}}$, this concludes the proof. 
\end{proof}

\end{document}